\documentclass[aps,twocolumn,10pt,prb]{revtex4-1} 

\usepackage[usenames]{xcolor}

\usepackage{titlesec}
\definecolor{natureblue}{HTML}{343796}
\titleformat{\section}[display]{\vspace{-1em}}{}{0pt}{\normalfont\sffamily\textbf}[\vspace{-1em}]
\usepackage{setspace}
\makeatletter
\def\frontmatter@abstractfont{\bf\sffamily}%
\def\frontmatter@title@format{\noindent\huge\sffamily}{}%
\def\frontmatter@authorformat{\vspace{1em}\noindent\large\sffamily}%
\def\frontmatter@affiliationfont{\vspace{1em}\normalsize\noindent\sffamily}%
\def\frontmatter@above@affiliation@script{\vspace{1em}\noindent}%
\def\frontmatter@makefnmark{}
\renewcommand*\frontmatter@date[2][\Dated@name]{\def\@date{}}%
\makeatother

\usepackage[english]{babel}
\usepackage[T1]{fontenc} 
\usepackage{times}

\usepackage{amsthm} 
\usepackage{amssymb} 
\usepackage{amsmath} 
\usepackage{mathbbol} 
\usepackage[normalem]{ulem}
\usepackage{graphicx}
\usepackage{paralist} 

\newtheorem{theorem}{Theorem}
\newtheorem{lemma}[theorem]{Lemma}
\newtheorem{corollary}[theorem]{Corollary}

\newtheorem{definition}[theorem]{Definition}

\newtheorem{standardbox}{Box} 
\usepackage{mathtools} 

\usepackage[colorlinks=true,citecolor=blue,linkcolor=magenta]{hyperref}

\makeatletter
\newcommand{\manuallabel}[2]{\def\@currentlabel{#2}\label{#1}}
\makeatother

\usepackage{tikz}
\newcommand{\tikzbox}[1]{
\begin{table}[t]
\normalsize
\begin{tikzpicture}
\node[draw, rounded corners = .5ex,fill=gray!4]{
  \begin{minipage}{.982\columnwidth}
  #1
  \end{minipage}
  };
\end{tikzpicture}
\end{table}
}

\makeatletter
\newcommand{\pushright}[1]{\ifmeasuring@#1\else\omit\hfill$\displaystyle#1$\fi\ignorespaces}
\makeatother

\newcommand{\eq}{equation}
\newcommand{\eqs}{equations}
\newcommand{\Eq}{Equation}
\newcommand{\Eqs}{Equations}
\newcommand{\fig}{Fig.}

\newcommand{\e}{\mathrm{e}}
\renewcommand{\i}{\mathrm{i}}
\newcommand{\1}{\mathbb{1}}
\newcommand{\0}{\mathbf{0}}
\DeclareMathOperator{\landauO}{\mathrm{O}}
\DeclareMathOperator{\Tr}{Tr}

\newcommand{\RR}{\mathbb{R}}
\newcommand{\NN}{\mathbb{N}}
\newcommand{\PP}{\mathbb{P}}
\renewcommand{\P}{\mathbf{P}}
\newcommand{\PPP}{\mathcal{P}}

\newcommand{\ket}[1]{\left|{#1}\right\rangle}

\newcommand{\bra}[1]{\left\langle{#1}\right|}

\newcommand{\ketbra}[2]{\ket{#1} \!\! \bra{#2}}

\newcommand{\norm}[1]{\Vert #1 \Vert}

\newcommand{\normb}[1]{\bigl\Vert #1 \bigr\Vert}

\newcommand{\ad}{^\dagger}
\newcommand{\argdot}{{ \cdot }}
\renewcommand{\ol}[1]{\overline{#1}}
\newcommand{\kw}[1]{\frac{1}{#1}}

\DeclareMathOperator{\vect}{\mathrm{vec}}

\newcommand{\ns}{C}
\newcommand{\nsA}{c}
\newcommand{\nCor}[1]{N_{#1}}
\newcommand{\epsilonmax}{\epsilon_{\mathrm{max}}}
\newcommand{\varepsilonmax}{\varepsilon_{\mathrm{max}}}

\newcommand{\smax}{s_{\mathrm{max}}}
\newcommand{\A}{\mathcal{A}}
\newcommand{\bphi}{\boldsymbol{\phi}}

\newcommand{\FSys}{F_\Sys}
\newcommand{\FT}{F_{\mathrm{T}}}
\newcommand{\Fest}{F^{(\vec{n})\ast}}
\newcommand{\Festzero}{F^{(0)\ast}}
\newcommand{\Festone}{F^{(n)\ast}}

\newcommand{\Fzero}{F^{(0)}}
\newcommand{\Fone}{F^{(n)}}
\newcommand{\Fn}{F^{(\vec{n})}}

\newcommand{\FSysn}{\FSys^{(\vec{n})}}
\newcommand{\FSysone}{{\FSys}^{(n)}}
\newcommand{\FSyszero}{{\FSys}^{(0)}}
\newcommand{\FSysestzero}{\FSys^{(0)\ast}}
\newcommand{\FSysestone}{\FSys^{(n)\ast}}
\newcommand{\rhot}{\varrho_{\mathrm{t}}}

\newcommand{\rhop}{\varrho_{\mathrm{p}}}

\newcommand{\rhosysp}{{\varrho_{\Sys}}_{\mathrm{p}}}

\newcommand{\rhosyst}{{\varrho_{\Sys}}_{\mathrm{t}}}
\newcommand{\trhop}{\tilde{\varrho}_{\mathrm{p}}}
\newcommand{\trhopn}{\tilde\varrho_{\mathrm{p},\vec n}}
\newcommand{\ntperp}{\tilde{n}^{\perp}}

\newcommand{\Epsilon}{\mathcal{E}}
\renewcommand{\S}{\mathcal{S}}
\newcommand{\CG}{\mathcal{C}_\mathrm{G}}
\newcommand{\CLON}{\mathcal{C}_\mathrm{LO}}
\newcommand{\CLPSG}{\mathcal{C}_\mathrm{LPSG}}
\newcommand{\CLPSLON}{\mathcal{C}_\mathrm{LPSLO}}
\newcommand{\MG}{\mathcal{M}_\mathrm{G}}
\newcommand{\MLON}{\mathcal{M}_\mathrm{LO}}
\newcommand{\MLPSG}{\mathcal{M}_\mathrm{LPSG}}
\newcommand{\MLPSLON}{\mathcal{M}_\mathrm{LPSLO}}

\renewcommand{\vec}[1]{\mathbf{#1}}
\renewcommand{\S}{\vec{S}} 
\newcommand{\Sys}{\mathcal{S}}
\renewcommand{\O}{\vec{O}} 
\renewcommand{\P}{\vec{P}} 
\renewcommand{\o}{\vec{o}} 
\renewcommand{\0}{\vec{0}} 
\renewcommand{\r}{\vec{r}}
\newcommand{\n}{\vec{n}}
\newcommand{\bfepsilon}{\boldsymbol{\epsilon}}
\newcommand{\D}{\vec{D}}
\newcommand{\bfgamma}{\boldsymbol{\gamma}}
\newcommand{\bfGamma}{\boldsymbol{\Gamma}}

\newcommand{\fu}{Dahlem Center for Complex Quantum Systems, Freie Universit\"{a}t Berlin, 14195 Berlin, Germany}
\newcommand{\icfo}{\mbox{ICFO-Institut de Ci\`encies Fot\`oniques, Mediterranean Technology Park, 08860 Castelldefels (Barcelona), Spain}}
\newcommand{\mpq}{\mbox{Max-Planck-Institut f\"ur Quantenoptik, Hans-Kopfermann-Str.\ 1, 85748 Garching, Germany}}

\newcommand{\titletext}{Reliable quantum certification for photonic quantum technologies}

\begin{document}
%
\title{\titletext}
\author{Leandro Aolita$^{1}$, Christian Gogolin$^\text{1,2,3}$, Martin Kliesch$^\text{1}$, and Jens Eisert$^\text{1}$}
\affiliation{\mbox{$^\text{1}$\fu}\\
  \noindent{}$^\text{2}$\icfo\\
  \noindent{}$^\text{3}$\mpq}

\begin{abstract}
A major roadblock for large-scale photonic quantum technologies is the lack of practical reliable certification tools. We introduce an experimentally friendly --- yet mathematically rigorous --- certification test for experimental preparations of arbitrary $m$-mode pure Gaussian states, pure non-Gaussian states generated by linear-optical circuits with $n$-boson Fock-basis states as inputs, and states of these two classes subsequently post-selected with local measurements on ancillary modes. The protocol is efficient in $m$ and the inverse post-selection success probability for all Gaussian states and all mentioned non-Gaussian states with constant $n$. We follow the mindset of an untrusted prover, who prepares the state, and a skeptic certifier, with classical computing and single-mode homodyne-detection capabilities only. No assumptions are made on the type of noise or capabilities of the prover. Our technique exploits an extremality-based fidelity bound whose estimation relies on non-Gaussian state 
nullifiers, which we introduce on the way as a byproduct result. The certification of many-mode photonic networks, as those used for photonic quantum simulations, boson samplers, and quantum metrology, is now within reach.
\end{abstract}

\maketitle


Many-body quantum devices promise exciting applications in ultra-precise quantum metrology \cite{Giovannetti}, quantum computing \cite{Nielsen-Chuang, BlattErrorCorrection,Barends14},  
and quantum simulators \cite{CZ12,AW12,BDN12,BR12,HTK12}.
In the quest for their large-scale realisation, impressive progress on a variety of quantum technologies has recently been made \cite{AW12,BDN12,BR12,HTK12}.
Among them, optical implementations play a key role.
For example, sophisticated manipulations of multi-qubit entangled states of up to 
eight parametrically down-converted photons \cite{Yao12,Huang12} 
have been demonstrated and continuous-variable entanglement among 60 stable \cite{Chen13} and up to 10000 flying \cite{Yokoyama13} modes has been verified in optical set-ups. In addition, small-sized simulations of BosonSampling \cite{E1,E2,E3,E4} and Anderson localisation in quantum walks \cite{Peruzzo10,Crespi13} 
have been performed with on-chip integrated linear-optical networks. 

This fast pace of advance, however, makes the problem of \emph{reliable 
certification} an increasingly pressing issue \cite{BS1, BS2, SpagnoloValidation,CarolanValidation,Tichy}.
From a practical viewpoint, further experimental progress on many-body quantum technologies is nowadays hindered by the lack of practical certification tools. At a fundamental level, certifying many-body quantum devices is ultimately about testing quantum mechanics in regimes where it has never been tested before.

Tomographic characterisation of quantum states requires the measurement of exponentially many observables. 
Compressed-sensing techniques \cite{Gro+10}
reduce, for states approximated by low-rank density matrices, the requirements significantly, 
but still demand exponentially many measurements. 
Efficient certification techniques, requiring only polynomially many measurements, for universal quantum computation\cite{Aharonov,Broadbent} and a restricted model of computation with one pure qubit\cite{Kapourniotis} exist in the form of quantum interactive proofs. However, these require either a fully fledged fault-tolerant universal quantum computer \cite{Aharonov,Broadbent} or an experimentally non-trivial measurement-based quantum device \cite{Kapourniotis}. In addition, these methods involve sequential interaction rounds with the device \cite{Aharonov,Broadbent,Kapourniotis}. In contrast, permutationally invariant tomography \cite{Perm}, Monte-Carlo fidelity estimation \cite{Flammia&Liu,Silva11,Flammia2}, and Clifford-circuit benchmarking techniques \cite{EmersonGambetta} provide experimentally friendly alternatives for the efficient certification of preparations of permutationally invariant \cite{Perm} and qubit stabiliser or W states \cite{Flammia&Liu,Silva11,Flammia2,EmersonGambetta}, respectively.
Nevertheless, none of these methods addresses continuous-variable systems, not even in Gaussian states.

\begin{figure*}[t!]
\centering
\includegraphics[width=.8\linewidth]{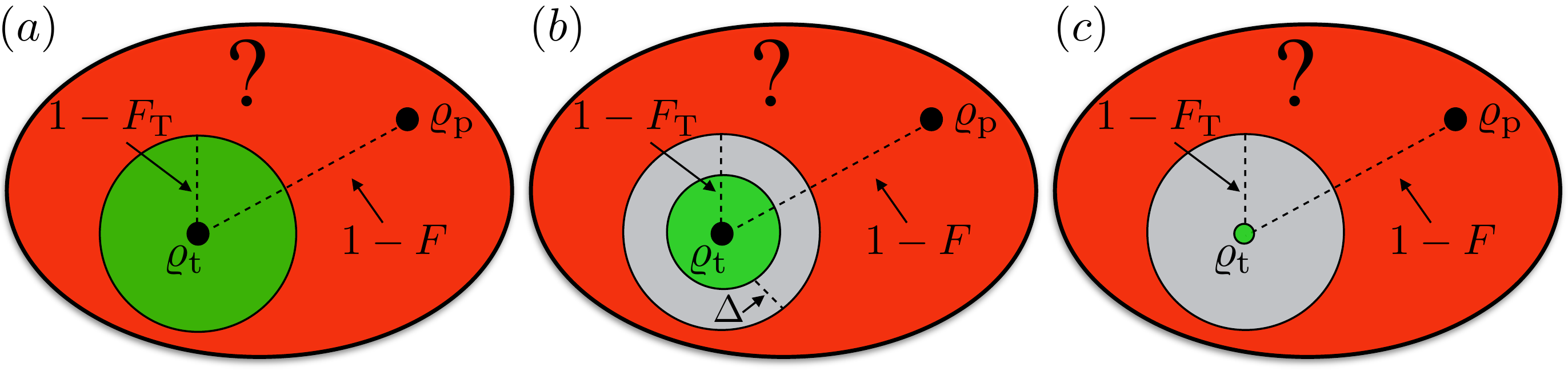}
\caption{{\bf Different certification paradigms}. ($a$) Naive approach: To certify an untrusted experimental preparation $\rhop$ of the target state ${\rhot}$, a certifier Arthur would like to run a statistical test that, for all $\rhop$, decides whether the fidelity $F$ between $\rhop$ and ${\rhot}$ is greater or equal than a pre-specified threshold $\FT<1$ (green region, accept), or smaller than it (red region, reject). However, due to the preparations at the boundary of the two regions and experimental uncertainties, a test able to make such a decision does not exist. ($b$) The ideal scenario: A more realistic certification notion is to ask that the test rejects every $\rhop$ for which $F<\FT$ (red region) and accepts every $\rhop$ for which $F\geq \FT+\Delta$ (green region), for some given $\Delta<1-\FT$. Here, a buffer region of width $\Delta$ (in grey) is introduced within which the behaviour of the test can be arbitrary, but, in return, the certification is now feasible. This type of certification is 
thus 
robust against experimental infidelities as large as $1-\FT-\Delta$. ($c$) The practical scenario: Finally, the least one can demand is that the test rejects every $\rhop$ for which $F<\FT$ (red region) and accepts at least $\rhot$ (green point). 
The former condition is sometimes called \emph{soundness} and the latter one \emph{completeness}.
Here, no acceptance is guaranteed for any $\rhop$ with $F\geq \FT$ (grey region) other than $\rhot$ itself, but any $\rhop$ accepted by the test necessarily features $F\geq \FT$. This certification notion is not robust against state deviations, but it can be more practical. In addition, in practice, the resulting tests succeed also in accepting many $\rhop\neq\rhot$ for which $F\geq \FT$.}
\label{fig:1}
\end{figure*}
Here, we introduce an experimentally friendly technique for the direct certification of continuous-variable
state preparations without estimating the prepared state itself.
First, we discuss intuitively and define rigorously reliable quantum-state certification tests. 
We do this for two notions of certification, differing in that in one of them robustness against preparation errors is mandatory. 
Then, we present a certification test, based on single-mode homodyne detection, for arbitrary $m$-mode pure Gaussian states, non-Gaussian states resulting from Gaussian unitary operators acting on Fock-basis states with $n$ photons, and states prepared by post-selecting states in either of the two classes with measurements on $a<m$ ancillary modes in arbitrary local bases. 
This covers, for instance, Gaussian quantum simulations such as those of refs.~\onlinecite{Chen13, Yokoyama13} as well as the non-Gaussian ones of refs.~\onlinecite{AW12,Yao12,Huang12, E1,E2,E3,E4,Peruzzo10,Crespi13}.
Furthermore, so-called de-Gaussified (photon-subtracted) Gaussian states~\cite{DellAnno07,Navarrete-Benlloch12,DellAnno13,Distillation} as well as all non-Gaussian states accessible to qumode-encoded qubit~\cite{KLM,KLM_Review} or finite-squeezing qumode~\cite{Menicucci06,Gu09} quantum computers also lie within the range of applicability of our method.
The protocol is efficient in $m$ and, for the cases with post-selection, in the inverse polynomial post-selection success probability, for all Gaussian states and all mentioned non-Gaussian states with constant $n$. 

With a high probability, our test rejects all experimental preparations with a fidelity with respect to the 
chosen target state lower than a desired threshold  and accepts if the preparation is sufficiently close to the target. That is, the protocol is robust against small preparation errors. 
We upper-bound the failure probability in terms of the number of experimental runs and calculate the necessary number of measurement settings.
Our method is built upon a fidelity lower bound, based on a natural extremality property, that is interesting in its own right. Finally, the experimental estimation of this bound relies on non-Gaussian state nullifiers, which we introduce on the way.

\section*{Results}
\label{sec:results}

We present our results in terms of photons propagating through optical networks, but our methods apply to any bosonic platform with equivalent dynamics. We consider a sceptic certifier, Arthur, with limited quantum capabilities, who wishes to ascertain whether an untrusted quantum prover, Merlin, presumably with more quantum capabilities, can indeed prepare certain quantum states that Arthur cannot. This mindset is reminiscent to that of quantum interactive-proof systems \cite{Aharonov,Broadbent,Kapourniotis} of computer science, but our method has the advantage that no interaction apart from the measurements of the certifier on the single-run experimental preparations from the prover is required.

In particular, we consider the situation where Merlin possesses at least a network of active single-mode squeezers and displacers as well as passive beam-splitters and phase-shifters, sufficient to efficiently implement any $m$-mode Gaussian unitary \cite{Reck,Braunstein,RMPReview,MyReview}, plus single-photon sources. Arthur's resources, in contrast, are restricted to classical computational power augmented with single-mode measurements. With that, he can characterise each of his single-mode measurement channels up to any desired constant precision. The task is for Merlin to provide him with copies of an $m$-mode pure \emph{target state} $\rhot$ of Arthur's choice. We assume that Merlin follows independent and identical state-preparation procedures on each experimental run, described by the density matrix $\rhop$. We refer to $\rhop$ as a \emph{preparation} of the target state $\rhot$. His preparation is unavoidably subject to imperfections and he might even be 
dishonest and try to trick Arthur. Thus, Arthur would like to 
run a test, with his own measurement devices, to \emph{certify} whether $\rhop$ is indeed a bona fide preparation of ${\rhot}$.

To measure how good a preparation $\rhop$ of ${\rhot}$ is, we use the fidelity between ${\rhop}$ and $\rhot$, defined as 
\begin{equation}
  \label{eq:fidelity}
  F\coloneqq F({\rhot},\rhop):=\Tr\big[(\sqrt{{\rhot}}\rhop^{\dagger}\sqrt{{\rhot}})^{1/2}\big]^2=\Tr\big[{\rhot}\rhop\big],
\end{equation}
where the last equality holds because $\rhot$ is assumed to be pure. As we see below, our measurement schemes directly estimate fidelities. However, all our results can also be adapted to the trace distance $D\coloneqq D({\rhot},\rhop)$, which can be defined via the 1-norm distance in state space as $D({\rhot},\rhop) \coloneqq \Tr[ | {\rhot}-\rhop | ]/2$. This is due to the fact that $D$ can be bounded from both sides in terms of $F$ through the well-known inequalities $1-F^2\leq D\leq\sqrt{1-F^2}$,
where the first inequality holds because ${\rhot}$ is pure. 

Let us first discuss what properties an experimental test must fulfil to qualify as a state certification protocol. Different certification paradigms are schematically represented in \fig~\ref{fig:1}. We start with the formal definition of certification in the sense of \fig~\ref{fig:1} (c). 
\begin{definition}[Quantum state certification] 
\label{def:certification}
Let $\FT<1$ be a threshold fidelity and $\alpha>0$ a maximal failure probability. A test, which takes as input a classical description of ${\rhot}$ and copies of a preparation $\rhop$ and outputs ``accept'' or ``reject'' is a \emph{certification test} for ${\rhot}$ if, with probability at least $1-\alpha$, it both rejects every $\rhop$ for which $F({\rhot},\rhop)< \FT$ and accepts $\rhop={\rhot}$. We say that any $\rhop$ accepted by such a test is a \emph{certified preparation} of ${\rhot}$.
\end{definition}

\begin{figure*}[t!]
\centering
\includegraphics{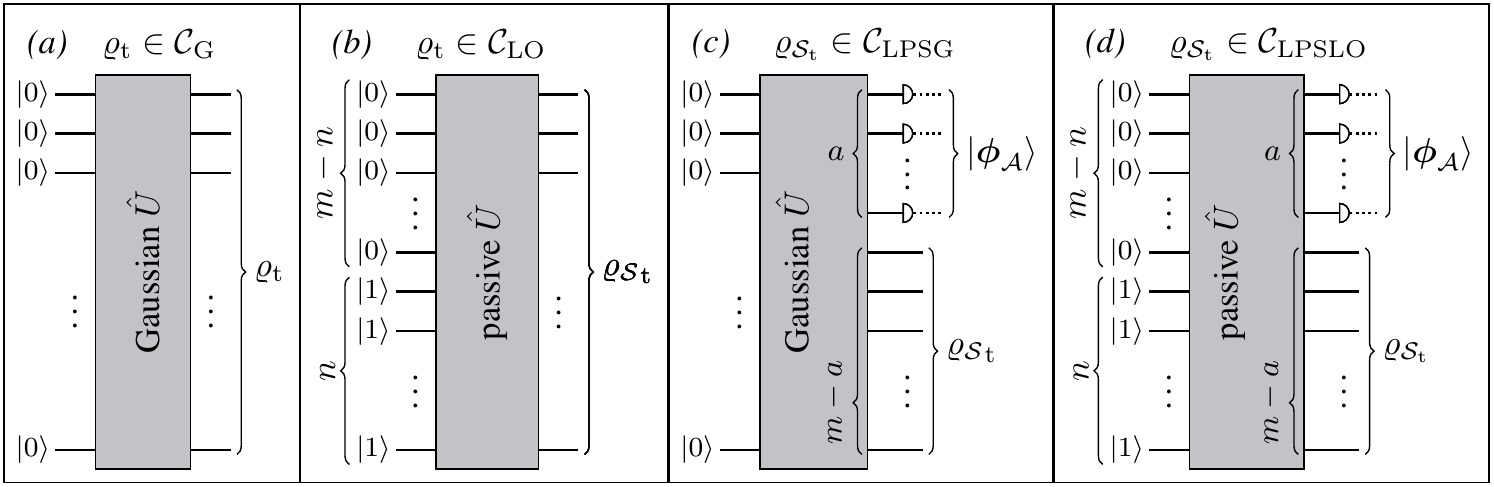}
\caption{{\bf Classes of target states}. 
($a$) $\CG$ is the class composed of all $m$-mode pure Gaussian states. These can be prepared by applying an arbitrary Gaussian unitary  $\hat{U}$ (possibly involving multi-mode squeezing) to the $m$-mode vacuum state $\ket{\0}$. ($b$) The class $\CLON$ includes all $m$-mode pure non-Gaussian states produced at the output of an arbitrary linear-optical network, which implements a passive Gaussian unitary  $\hat{U}$ (without squeezing), with the Fock-basis state $\ket{\mathbf{1}_n}$ containing one photon in each of the first $n$ modes and zero in the remaining $m-n$ ones as input. As the order of the modes is arbitrary, choosing the first 
$n$ modes as the populated ones does not constitute a restriction. ($c$) The third class,  $\CLPSG$, encompasses all $(m-a)$-mode pure non-Gaussian states obtained by projecting a subset $\A$ of $a<m$ modes of an $m$-mode pure Gaussian state $\rhot\in\CG$ onto an arbitrary pure product state $\ket{\bphi}_\A$. In practice, this is done probabilistically by measuring $\A$ in a local basis that contains $\ket{\bphi}_\A$ and post-selecting only the events in which $\ket{\bphi}_\A$ is measured. Thus, the $a$ modes in $\A$ are used as ancillas, whereas the effective system is given by the subset $\Sys$ containing the other $m-a$ modes, which carries the final target state. For concreteness, but without any loss of generality, in the plot, the ancillary modes are chosen to be the last $a$ ones.
($d$) Analogously, the class $\CLPSLON$ is that of all $(m-a)$-mode pure non-Gaussian states obtained by projecting the ancillary modes of an $m$-mode pure linear-optical network state $\rhot\in\CLON$ onto an arbitrary pure product state $\ket{\bphi}_\A$. These four classes cover the target states considered in the vast majority of quantum photonic experiments.}
\label{fig:2}
\end{figure*}

To specify the target states we need to introduce some notation.
We denote $m$-mode \emph{Fock basis states} by $\ket{\n}$, with $\n\coloneqq(n_1,n_2,\hdots ,n_m)$ being
the sequence of photon numbers $n_j\geq0$ in each mode $j\in [m]$, where the short-hand notation $[m]\coloneqq\{1,2, \hdots , m\}$ is introduced, and call $n\coloneqq\sum_{j=1}^m n_j$ the \emph{total input photon number}.
In particular, we will pay special attention to Fock basis states $\ket{\mathbf{1}_n}$ with exactly one photon in each of the first $n$ modes and the vacuum in the remaining $m-n$ ones, i.e., those for which $\n= \mathbf{1}_n$, with
\begin{equation}
  \mathbf{1}_n \coloneqq (\underbrace{1,\dots,1}_{n\ \text{times}},\underbrace{0,\dots,0}_{m-n\ \text{times}}).
\end{equation}
Note that $\ket{\mathbf{1}_0}$ is the Gaussian \emph{vacuum state} $\ket{\0}$.
We denote the \emph{photon number operator} corresponding to mode $j$ by $\hat n_j$ and the \emph{total photon number operator} by $\hat n \coloneqq\sum_{j=1}^m \hat n_j$.

In addition, for post-selected target states, we denote by $\A\coloneqq\{\A_j\}_{j\in[a]}$, where each element $\A_j\in[m]$ labels a different mode, the subset of $a\coloneqq|\A|<m$ modes on which the post-selection measurements are made. We then identify the remaining $m-a$ modes as the system subset $\Sys$, which carries the post-selected target state $\rhosyst$. The subindex $\Sys$ emphasises that $\rhosyst$ represents an $(m-a)$-mode post-selected target state and distinguishes it from $m$ mode target states without post-selection, which we denote simply as $\rhot$.
We denote by
$\ket{\bphi}_\A\coloneqq\ket{\phi_1}_{\A_1}\ket{\phi_2}_{\A_2}\hdots\ket{\phi_a}_{\A_a}$, with $\{\ket{\phi_j}_{\A_j}\}_{j\in a}$ an arbitrary pure normalised state of mode $\A_j$,  an $a$-mode product state on the modes $\A$. 
We use the short-hand notations $\bra{\bphi}_\A\rhot\ket{\bphi}_\A\coloneqq\Tr_\A\left[\rhot (\1_{\Sys}\otimes\ket{\bphi}_\A\bra{\bphi}_\A )\right]$, where $\Tr_\A$ indicates partial trace over the Fock space of $\A$, $\1_\Sys$ denotes the identity on $\Sys$, and $\mathbb{P}(\bphi_\A|\rhot)\coloneqq\Tr\left[\bra{\bphi}_\A\rhot\ket{\bphi}_\A\right]$ is the \emph{post-selection success probability}, i.e., the probability of measuring $\ket{\bphi}_\A$ in a projective measurement on $\A$.
Without loss of generality, we consider throughout only the non-trivial case $\mathbb{P}(\bphi_\A|\rhot)\neq0$.
Thus, we consider exclusively post-selected target states of the form
\begin{align}
   \rhosyst \coloneqq\frac{\bra{\bphi}_\A\rhot\ket{\bphi}_\A}{\mathbb{P}(\bphi_\A|\rhot)}.
   \label{eq:general_PS_state}
\end{align} 

With the notation introduced, we derive our results for:

\begin{compactenum}[1)]
\item Arbitrary $m$-mode pure \emph{Gaussian states}, given by the class
 \begin{align}
    \CG &
    \coloneqq \{\rhot=\hat{U}\ketbra{\0}{\0}\hat{U}^{\dagger} \colon \hat{U} \text{ Gaussian unitary} \}, \label{eq:SG} \\
\intertext{\item $m$-mode pure \emph{linear-optical network states} from the class}
    \CLON &\coloneqq \{\rhot=\hat{U}\ketbra{\mathbf{1}_n}{\mathbf{1}_n}\hat{U}^{\dagger} \colon \hat{U} \text{ passive unitary} \}  , \label{eq:SLON}
\intertext{\item arbitrary $(m-a)$-mode pure \emph{locally post-selected Gaussian states}, given by the class}
    \CLPSG &\coloneqq \left\{\rhosyst\colon\rhot\in\CG\right\}  ,\label{eq:LPSG}
\intertext{\item and $(m-a)$-mode pure \emph{locally post-selected linear-optical network states} from the class}
    \CLPSLON &\coloneqq \left\{\rhosyst\colon\rhot\in\CLON\right\}  . \label{eq:LPSLON}
\end{align}
\end{compactenum}

The class $\CG$ is crucial within the realm of ``continuous-variable'' quantum optics and quantum information processing. It encompasses, for instance, ``twin-beam" (two-mode squeezed vacuum) states under passive networks, which are used to simulate, upon coincidence detection, multi-qubit states \cite{AW12}. 
The class $\CLON$ includes all the settings sometimes referred to as ``discrete variable'' linear-optical networks.
This class covers, among others, the targets of several recent experimental simulations with on-chip integrated linear-optical networks \cite{E1,E2,E3,E4,Peruzzo10,Crespi13}.
The third class, $\CLPSG$, is the one of locally post-selected Gaussian states. This class includes crucial non-Gaussian resources for quantum information and quantum optics. For instance, when the post-selection is in the Fock basis, it encompasses de-Gaussified photon-subtracted squeezed Gaussian states~\cite{DellAnno07,Navarrete-Benlloch12,DellAnno13,Distillation}. Furthermore, if apart from Fock-basis measurements, the post-selection uses also quadrature homodyne measurements, $\CLPSG$ contains all the states accessible to finite-squeezing cluster-state qumode quantum computers~\cite{Menicucci06,Gu09}.
The last class, $\CLPSLON$, of locally post-selected linear-optical network states, covers, for the case where the post-selection is in the Fock basis and $n$ is proportional to $m$, all the states prepared  
by probabilistic schemes of the type of refs. \onlinecite{KLM,KLM_Review} for universal qumode-encoded qubit quantum computation. 
Naturally, $\CLPSLON$ also includes both photon -added or -subtracted linear-optical network states.

The basis of the our certification scheme is a technique for the estimation of the quantity
\begin{equation}
\label{eq:eq:firstineqgeneric}
\Fone\coloneqq 1 - \left\langle (\hat{n}-n)\prod_{j=1}^{n}\hat{n}_{j}\right\rangle_{\hat{U}^{\dagger}\rhop\hat{U}},
\end{equation}
with $n$ the total input photon number. 
As shown in 
the Methods section, for all target states $\rhot \in \CG \cup \CLON$, $\Fone$ is a lower bound on the fidelity $F$ and, moreover, $\Fone=F=1$ if $\rhop=\rhot$
(see also Methods and Section~\ref{sec:PS_fidelity_bound} of the Supplementary Information (SI) for analogous bounds for the post-selected target states).
This bound is a consequence of a natural extremality notion:
the smaller the expectation value 
$\big\langle (\hat{n}-n)\prod_{j=1}^{n}\hat{n}_{j}\big\rangle_{\hat{U}^{\dagger}\rhop\hat{U}}$
is, the closer are $\ket{\mathbf{1}_n}\bra{\mathbf{1}_n}$ and $\hat{U}^{\dagger}\rhop\hat{U}$ and, therefore, the closer are the preparation $\rhop$ and the target state $\rhot$.
Our test $\mathcal{T}$, summarised in Box~\ref{box:test}, yields an estimate $\Festone$ of $\Fone$. If $\Festone$ is sufficiently above the threshold $\FT$, the preparation $\rhop$ is accepted.
Otherwise it is rejected.
The estimate $\Festone$ is obtained via a measurement scheme that depends on the specific target state.
In the Gaussian case $n=0$ the measurement scheme $\MG$ can be used, while linear-optical network states with $n>0$ require the scheme $\MLON$. $\MG$ and $\MLON$ are both summarised in the Methods section and described in detail in Boxes~\ref{box:measurement_scheme} and \ref{box:measurement_scheme_1n}, respectively, in Section~\ref{sec:post_selected_Targets} in the SI.
In addition, in Section~\ref{sec:PS_cert_test} of the SI we adapt $\mathcal{T}$ to post-selected target states $\rhosyst\in\CLPSG\cup\CLPSLON$, and provide the corresponding adapted measurement schemes in Section~\ref{sec:PS_measurement_scheme} of the SI. 

\tikzbox{
  \begin{standardbox}[Certification test $\mathcal{T}$]
  \label{box:test}
  \noindent\begin{compactenum}[1)] 
  \item Arthur chooses a threshold fidelity $\FT<1$, a maximal failure probability $\alpha>0$, and an estimation error $0<\varepsilon\leq (1-\FT)/2$. 
  \item Arthur provides Merlin with the classical specification $n$, $\S$, and $\vec{x}$ of the target state $\rhot$ and requests a sufficient number of copies of it. 
  \item If $n=0$,
  Arthur measures $2 m \kappa$ two-mode correlations and $2 m$ single-mode expectation values specified by the measurement scheme $\MG$ (see the Methods section), which can be done with $m+3$ single-mode homodyne settings.
  \\
  If $n>0$,
  he measures $\landauO\big(m(4d^2+1)^{n}\big)$  multi-body correlators, each one involving  between $1$ and $2n+1$ modes, specified by the measurement scheme $\MLON$ (see the Methods section), which can be done with at most $\binom m n  2^{n+1}$ single-mode homodyne settings. 
  \item \label{itme:1.4_MAIN}
  By classical post-processing (see the Methods section), he obtains a fidelity estimate $\Festone$ such that $\Festone\in [\Fone-\varepsilon,\Fone+\varepsilon]$ with probability at least $1-\alpha$, where $\Fone$ is the lower bound to $F$ given by expression \eqref{eq:eq:firstineqgeneric}.
  \item \label{item:accept_reject_condition_MAIN}
  If $\Festone<\FT+\varepsilon$, he rejects. Otherwise, he accepts.
  \end{compactenum} 
  \end{standardbox}
}

Our theorems guarantee that the test from Box~\ref{box:test} is indeed a certification test and give a bound on the scaling of the number of samples that are needed for the test.
In order to state them we introduce some notation related to mode space descriptions of linear-optical networks first.
Any Gaussian unitary transformation $\hat U$ on Hilbert space can be represented by an affine symplectic transformation in mode space, i.e., by a \emph{symplectic matrix} $\S\in\mathrm{Sp}(2m,\mathbb{R})$ followed by a phase-space \emph{displacement} $\vec{x}\in\mathbb{R}^{2m}$  (see \eq~\eqref{eq:coordinatetransformdirect} in the Methods section), where the real
symplectic group $\mathrm{Sp}(2m,\mathbb{R})$ contains all real $2m\times2m$ matrices that preserve the canonical phase-space commutation relations \cite{RMPReview,MyReview}.
By virtue of the Euler decomposition \cite{Braunstein,RMPReview}, $\S$ can be implemented with single-mode squeezing operations and passive mode transformations.
We denote  the \emph{maximum single-mode squeezing} of $\S$ by $\smax$ and define the \emph{mode range} $d \leq m$ to be the maximal number of input modes to which each output mode is coupled (for details see Section~\ref{sec:measurement_scheme} 
of the SI).
Also, it will be useful to define
\begin{equation} 
\label{eq:kappa_def}
  \kappa\coloneqq2\min\{d^2,m\} .
\end{equation}
The displacement $\vec{x}$ can be implemented by a single-mode displacer at each mode $j\in[m]$, with amplitude $(x_{2j-1},x_{2j})$, where $x_{k}$, for $k\in[2m]$, is the $k$-th component of $\vec{x}$.
The \emph{vector $2$-norm} is denoted by $\norm{\cdot}_2$, i.e., $\norm{\vec x}_2 \coloneqq\bigl(\sum_{k=1}^{2m}x^2_{k}\bigr)^{1/2}$.

We take $\sigma_i$ to be a uniform upper bound on the variances of any product of $i$ \emph{phase space quadratures} 
in the state $\rhop$.
If $\rhop$ is Gaussian, $\sigma_{1}$ and $\sigma_{2}$ are functions of the single mode squeezing parameters of $\rhop$.
In addition, we call $\sigma_{\leq i} \coloneqq \max_{k \leq i}\{\sigma_k\}$ the \emph{maximal $i$-th variance} of $\rhop$.
Finally, we use the Landau symbol $\landauO$ to denote asymptotic upper bounds.

\begin{theorem}[Quantum certification of Gaussian states] 
\label{thm:certification_gaussian}
Let $\FT<1$ be a threshold fidelity, $\alpha>0$ a maximal failure probability, and $0<\varepsilon\leq(1-\FT)/2$ an estimation error. Let $\rhot\in \CG$ have maximum single-mode squeezing $\smax\geq1$, mode range $d\leq m$, and displacement $\vec{x}$. 
Test $\mathcal{T}$ from Box~\ref{box:test} is a certification test for $\rhot$ and requires at most 
\begin{equation}
\label{eq:numb_copies_bound_Gaussian}
\landauO\!\left( 
 \frac{\smax^4\left(2\sigma_1^2 \norm{\vec x}_2^2 m^3 +  \sigma_2^2 \kappa^3 m^4\right)}
  {\varepsilon^2  \ln(1/(1-\alpha))} 
 \right) 
\end{equation} 
copies of a preparation $\rhop$ with first and second variance bounds $\sigma_{1}>0$ and $\sigma_{2}>0$, respectively.
\end{theorem}

\begin{theorem}[Quantum certification of linear-optical network states] \label{thm:certification_1n}
Let $\FT<1$ be a threshold fidelity, $\alpha>0$ a maximal failure probability, and $0<\varepsilon\leq(1-\FT)/2$ an estimation error. Let $\rhot\in\CLON$ have mode range $d\leq m$. Test $\mathcal{T}$ from Box~\ref{box:test} is a certification test for $\rhot$ and requires at most  
\begin{equation}
\label{eq:numb_copies_bound}
\landauO\!\left(
 \frac{\sigma_{\leq 2 (n+1)}^2 m^4 (\lambda\, d^{6}\, n\, m)^{n}}{\varepsilon^2 \ln(1/(1-\alpha))}
 \right)
 \end{equation}
copies of a preparation $\rhop$ with maximal $2(n+1)$-th variance $\sigma_{\leq2(n+1)}$, where $\lambda > 0$ is an absolute constant.
\end{theorem}

The proofs of all our theorems are provided in the SI. The treatments of the classes $\CLPSG$ or $\CLPSLON$ follow as corollaries of Theorems \ref{thm:certification_gaussian} and \ref{thm:certification_1n}, respectively, and are also provided in the SI (see Section~\ref{sec:PS_Corollary} there). Expressions \eqref{eq:numb_copies_bound_Gaussian} and \eqref{eq:numb_copies_bound} are highly simplified upper bounds on the total number of copies of $\rhop$ that $\mathcal{T}$ requires. For more precise expressions see Lemmas~\ref{lem:fidelity_bound_estimation} and~\ref{lem:fidelity_bound_estimation_1n} of the SI.
Note that neither of the two theorems requires any energy cut-off or phase-space truncation. While our bound~\eqref{eq:numb_copies_bound} is inefficient in $n$, both for the Gaussian and linear-optical cases, the number of copies of $\rhop$ scales polynomially with all other parameters, in particular with $m$. Thus, arbitrary $m$-mode target states from the classes $\CG$ and $\CLON$ with constant $n$, are certified by $\mathcal{T}$ efficiently.

Interestingly, since states in $\CLON$ in general display negative Wigner functions, sampling from their measurement probability distributions cannot be efficiently done by the available classical sampling methods \cite{Mari, Veitch}. 
Furthermore, for Fock-state measurements, these distributions define BosonSampling, for which hardness results exist\cite{AA13} for $m$ asymptotically lower-bounded by $n^5$.

Also, note that there are no restrictions on $\rhop$ except that, in practice, to apply the theorems, one needs bounds on $\sigma_{1}$, $\sigma_{2}$, and $\sigma_{\leq2(n+1)}$. 
These variances are properties of $\rhop$ and are therefore a priori unknown to Arthur. 
However, he can reasonably estimate them from his measurements. 
Note that, for random variables that can take any real value, assuming that the variances are bounded is a fundamental and unavoidable assumption to make estimations from samples; and it is one that can be contrasted with the measurement results.

To end up with, we consider certification in the robust sense of \fig~\ref{fig:1} (b):
\begin{definition}[Robust quantum state certification] 
\label{def:robust_certification}
Let $\FT<1$ be a threshold fidelity, $\alpha>0$ a maximal failure probability, and $\Delta<1-\FT$ a fidelity gap. A test, which takes as input a classical description of the target state ${\rhot}$ and copies of a preparation $\rhop$ and outputs ``accept'' or ``reject'' is a \emph{robust certification test} for ${\rhot}$ if, with probability at least $1-\alpha$, it both rejects every $\rhop$ for which $F({\rhot},\rhop)< \FT$ and accepts every $\rhop$ for which $F(\rhot,\rhop)\geq \FT+\Delta$. We say that any $\rhop$ accepted by such a test is a \emph{certified preparation} of ${\rhot}$. 
\end{definition}
\noindent This definition is more stringent than Definition~\ref{def:certification} in that it guarantees that preparations sufficiently close to ${\rhot}$ are necessarily accepted, rendering the certification robust against state deviations with infidelities as large as $1-(\FT+\Delta)$. We show below that our test $\mathcal{T}$ from Box~\ref{box:test} is actually a robust certification test.

To this end, we first write $\rhop$ as
\begin{equation}
\label{eq:rhop_in_terms_rhot}
\rhop=F\rhot+(1-F)\rhot^{\perp},
\end{equation}
where $\rhot^{\perp}$ is an operator orthogonal to $\rhot$ with respect to the Hilbert-Schmidt inner product, i.e., such that $\Tr[{\rhot}\,\rhot^{\perp}]=0$.
As $\rhot$ is assumed to be pure, it follows immediately that $\rhot^{\perp}$ is actually a state.
In fact, multiplying by $\rhot$ and taking the trace on both sides of \eq~\eqref{eq:rhop_in_terms_rhot}, one readily sees that the decomposition \eqref{eq:rhop_in_terms_rhot} is just another way to express the fidelity \eqref{eq:fidelity}. We define the \emph{photon mismatch} $\tilde{n}^{\perp}$ between $\rhot$ and $\rhop$ as 
\begin{equation}
\label{eq:def_phononmismatch}
\ntperp\coloneqq
\langle (\hat{n}-n)\prod_{j=1}^{n}\hat{n}_{j}\bigr \rangle_{\hat{U}^{\dagger}\rhot^{\perp}\hat{U}} .
\end{equation}
The photon mismatch gives the expectation value that Arthur would obtain if he had access to $\rhot^{\perp}$, applied the inverse of Merlin's network to it, and then measured $(\hat{n}-n) \prod_{j=1}^{n}\hat{n}_{j}$. For the ideal case $\rhop=\rhot$, it clearly holds that $\ntperp=0$.

\begin{theorem}[Robust quantum certification] 
\label{thm:strong_certification}
Under the same conditions as in Theorems~\ref{thm:certification_gaussian} and \ref{thm:certification_1n}, test $\mathcal{T}$ from Box~\ref{box:test} is a robust certification test with fidelity gap 
\begin{equation}
\label{eq:delta}
\Delta\coloneqq\max\left\{\frac{2\varepsilon+(1-\FT)(\tilde{n}^{\perp}-1)}{\tilde{n}^{\perp}},2\varepsilon\right\},
\end{equation}
where $\ntperp$ is the photon mismatch. 
\end{theorem}

\noindent As expected, the gap cannot be smaller than twice the estimation error for any photon mismatch. Notice also that in the limit $\ntperp\to\infty$ it holds that $\Delta\to1-\FT$, so that the certification becomes less robust with increasing $\ntperp$. 
As $\ntperp$ decreases from infinity to one, the gap decreases to its minimal value $\Delta=2\varepsilon$, where it remains for all $0\leq\ntperp\leq1$. 
We emphasise that $\ntperp$ depends on $\rhot^\perp$. Thus it cannot be directly estimated from measurements on $\rhop$ alone. However, for any $\ntperp<\infty$, Theorem~\ref{thm:strong_certification} guarantees the existence of an entire region of states around $\rhot$ that are rightfully accepted.
Furthermore, in the experimentally relevant situations, $\ntperp$ is expected to be small.  In this case, Theorem~\ref{thm:strong_certification} 
provides a lower bound on the size of the region of accepted states.

Finally, a statement equivalent to Theorem \ref{thm:strong_certification} for target states $\rhosyst\in\CLPSG\cup\CLPSLON$ follows as an immediate corollary of it and is presented in Section~\ref{sec:PS_Corollary_strong} in the SI.

\section*{Discussion}
\label{Conclusion}
Large-scale photonic quantum technologies promise important scientific advances and technological applications.
So far, considerably more effort has been put into their realisation than into the verification of their correct functioning and reliability. This imposes  a  serious obstacle for further experimental advance, specifically in the light of the speed at which progress towards many-mode architectures takes place.
Here, we  have presented a practical reliable certification tool 
for a broad family of multi-mode bosonic quantum technologies.

We have proven theorems that upper-bound the number of experimental runs sufficient for our protocol to be a certification test. Our theorems provide large-deviation bounds from a simple extremality-based fidelity lower-bound that is interesting in its own right. 
Importantly, our theorems hold only for statistical errors, but the stability analyses on which they rely (see Lemmas~\ref{lem:stability} and \ref{lem:stability_1n} in the SI) holds regardless of the nature of the errors. As a matter of fact, in Section~\ref{sec:systematic_errors} 
in the SI, we show that our fidelity estimates are robust also against systematic errors.

From a more practical viewpoint, our test allows one to certify the state preparations of most current optical experiments, in both the ``continuous-variable'' and the ``discrete-variable'' setting. This is achieved under the minimal possible assumptions: namely, only that the variances of the measurement outcomes are finite. Thus, the certification is as unconditional as the fundamental laws of 
statistics allow.
In particular, no assumption on the type of quantum noise is made.
Despite the rigorous bounds on the estimation errors and failure probabilities, our methods are both experimentally friendly and resource efficient.

Notably, our test can for instance be applied to the 
certification of optical circuits of the type used in BosonSampling:
There, $m$-mode Fock-basis states of $n$ photons are subjected to a linear-optical network described by a random unitary $\hat U$ drawn from the Haar measure \cite{AA13} and, subsequently, each output mode is measured in the Fock basis. 
While the question of the certification of the classical outcomes of such samplers without assumptions on the device is still largely open \cite{BS1,BS2},  with the methods described here
the pre-measurement non-Gaussian quantum outputs of BosonSampling devices \cite{E1,E2,E3,E4} can be certified reliably and, for constant $n$, even
efficiently. In this sense, this work goes significantly beyond previously proposed schemes to rule out particular cheating strategies by the prover \cite{BS2,SpagnoloValidation,CarolanValidation,Tichy}. 
Furthermore, a variety of non-Gaussian states paradigmatic in quantum optics and quantum information    are also covered by our protocol (see Section~\ref{sec:post_selected_Targets} in the SI for details). These include, for instance, de-Gaussified photon-subtracted multi-mode Gaussian states~\cite{DellAnno07,Navarrete-Benlloch12,DellAnno13,Distillation},  
multi-mode squeezed Gaussian states post-selected through photon-number or quadrature measurements, as in finite-squeezing cluster-state qumode quantum computers~\cite{Menicucci06,Gu09}, and linear-optical network outputs post-selected though photon-number measurements, ranging from photon -added or -subtracted linear-optical network states to all the states preparable with Knill-Laflamme-Milburn-like schemes~\cite{KLM,KLM_Review}. 
For all such states, our test is efficient in the inverse post-selection success probability $1/\mathbb{P}(\bphi_\A|\rhot)$.

The present method constitutes a step forward in the field of photonic quantum certification, with potential implications on the certification of other many-body quantum-information technologies.
Apart from that of BosonSamplers and optical schemes with post-selection, the efficient and reliable certification of large-scale photonic networks as those used, for instance, for multi-mode Gaussian quantum-information processing \cite{Chen13, Yokoyama13}, non-Gaussian Anderson-localisation simulations \cite{Peruzzo10,Crespi13}, and quantum metrology \cite{Giovannetti}, with a constant number of input photons, is now within reach.

\section*{Methods}
\textbf{Fidelity lower bound}. In this section, we formalise the extremality notion and derive a lower bound on the fidelity $F$. 
All target states are of the form 
\begin{equation}
\label{eq:general_target}
\rhot = \hat{U}\ketbra{\vec{n}}{\vec{n}}\hat{U}^{\dagger},
\end{equation}
where $\hat{U}$ is an arbitrary Gaussian unitary and $\ket{\vec{n}}$ an arbitrary Fock-basis state. First, we derive a general fidelity lower bound and then consider the linear-optical $\rhot \in \CLON$ and Gaussian  $\rhot \in \CG$ cases separately. 
Analogous bounds for the post-selected target states are provided further below in the Measurement Scheme and Section~\ref{sec:PS_fidelity_bound} of the SI. 

We start recalling that 
\begin{equation}
\label{number}
  \ket{\n} = \prod_{j=1}^m \frac{1}{\sqrt{n_j!}}  (\hat{a}\ad_j)^{n_j}  \ket{\mathbf{0}},
\end{equation}
where $a\ad_j$ is the creation operator of the $j$-th mode. Its Hermitian conjugated $\hat{a}_{j}$ is the corresponding annihilation operator. These operators satisfy $[\hat{a}_j,\hat{a}^{\dagger}_{j'}]=\delta_{j,j'}$, where $\delta_{j,j'}$ denotes the \emph{Kronecker delta} of $j$ and $j'$, and $\hat{n}_j=\hat{a}^{\dagger}_{j}\hat{a}_{j}$, for all $j,j'\in [m]$. 
The fidelity \eqref{eq:fidelity} can be written as $F= F(\ket{\n}\bra{\n},\tilde{\varrho}_{\mathrm{p}})$, where $\trhop \coloneqq \hat{U}^{\dagger}\rhop \hat{U}$ is the Heisenberg representation of $\varrho_{\mathrm{p}}$ with respect to $\hat{U}^{\dagger}$. 
With this, \eq~\eqref{number}, and the cyclicality property of the trace, we obtain that
\begin{equation}
 \label{firstFidelityNG}
 F=\Tr\left[\ketbra{\vec 0}{\vec 0} \trhopn \right]
 = F(\ketbra{\vec 0}{\vec 0}, \trhopn) ,
\end{equation}
where
\begin{equation} \label{eq:trhopn_def}
  \trhopn
  \coloneqq
  \prod_{j'=1}^m \frac{1}{\sqrt{n_{j'}!}}  (\hat{a}_{j'})^{n_{j'}} \tilde{\varrho}_{\mathrm{p}}\prod_{j=1}^m \frac{1}{\sqrt{n_j!}}  (\hat{a}\ad_j)^{n_j} .
\end{equation}

To lower-bound $F(\ketbra{\vec 0}{\vec 0}, \trhopn)$,
we consider the average total photon-number 
$\langle \hat{n}\rangle_{\trhopn}
\coloneqq 
\Tr[\hat{n}\trhopn]$ of $\trhopn$. 
We write $\1$ for the \emph{identity operator}.
From the facts $\1-\ketbra{\mathbf{0}}{\mathbf{0}}\leq\hat{n}$ and $\trhopn\geq0$, it follows that
\begin{align}
\label{eq:firstineq}
\langle \hat{n}\rangle_{\trhopn}&=\Tr\left[\sum_{\n} n \ketbra{\n}{\n}\trhopn\right]\nonumber \\
&\geq\Tr\left[(\1-\ketbra{\mathbf{0}}{\mathbf{0}})\trhopn\right]\nonumber \\
&=1-F 
\end{align}
and hence,
\begin{equation}
\label{eq:Fn_def}
F\geq \Fn \coloneqq1-\langle \hat{n}\rangle_{\trhopn}
.
\end{equation}
This bound justifies the  natural extremality intuition mentioned: The lower the average number of photons of $\trhopn$ is, the closer to the vacuum it must be and, therefore, the closer $\rhop$ to $\rhot$. Notice that, for $\rhop=\rhot$, the inequality in \eq~\eqref{eq:firstineq} becomes an equality and therefore bound \eqref{eq:Fn_def} is saturated, as announced earlier.

Next, we define the operator valued Pochhammer-Symbol
\begin{equation}
 \label{eq:pochhammer}
 p_{t}(\hat{n}_j) \coloneqq \hat{n}_j  (\hat{n}_j-1)  (\hat{n}_j-2)  \cdots  (\hat{n}_j-t) ,
\end{equation}
for any integer $t\geq 0$, and $p_{-1}(x) \coloneqq \1$. In Section~\ref{sec:proof_of_eq:general_fidelity_bound} 
in the SI we show that
\begin{subequations}
\label{eq:adagger}
\begin{align}
 \label{eq:adaggera}
 (\hat{a}\ad_j)^{n_j}\hat{n}_j(\hat{a}_j)^{n_j} &= p_{n_j}(\hat{n}_j), \\
\intertext{and}
 \label{eq:adaggerb}
 (\hat{a}\ad_j)^{n_j}(\hat{a}_j)^{n_j} &= p_{n_j-1}(\hat{n}_j).
\end{align}
\end{subequations}
Inserting \eq~\eqref{eq:trhopn_def} into \eq~\eqref{eq:Fn_def}, using the cyclicity property of the trace, grouping the operators of each mode together, using \eqs~\eqref{eq:adagger} and that $p_{t}(\hat{n}_j)=p_{t-1}(\hat{n}_j)\, (\hat{n}_j -t)$, we obtain the general fidelity lower bound
\begin{equation}
  \label{eq:general_fidelity_bound}
  F\geq \Fn = 1-\frac{1}{\n!}\bigg\langle (\hat n - n)  \prod_{j=1}^m p_{n_j-1}(\hat n_j)\bigg\rangle_{\tilde{\varrho}_{\text{p}}} ,
\end{equation}
where $\n!\coloneqq n_1! n_2! \hdots n_m!$. In order to specialise to the linear-optical case $\rhot \in \CLON$, we simply take $\vec{n} = \mathbf{1}_n$, i.e., $n_j=1$ for all $j\in[n]$ and $n_j=0$ otherwise. With this, $\Fn$ in \eq~\eqref{eq:general_fidelity_bound} simplifies to precisely the bound $\Fone$ in \eq~\eqref{eq:eq:firstineqgeneric}. Finally, to restrict it to the Gaussian case $\rhot \in \CG$, we take $n_j=0$ for all $j\in[m]$. This yields the particularly simple expression
\begin{equation}
\label{eq:boundfirst}
F\geq \Fzero \coloneqq1-\langle \hat{n}\rangle_{\tilde{\varrho}_{\mathrm{p}}}.
\end{equation}

Arthur does not have enough quantum capabilities to directly estimate $\langle \hat{n}\rangle_{\tilde{\varrho}_{\mathrm{p}}}$ by undoing the operation $\hat{U}$ on Merlin's outputs and then measuring $\hat{n}$ in the Fock state basis. However, we show in the next section that he can efficiently obtain $\langle \hat{n}\rangle_{\tilde{\varrho}_{\mathrm{p}}}$, as well as the expectation values in \eqs~\eqref{eq:general_fidelity_bound} and \eqref{eq:eq:firstineqgeneric}, from the results of single-mode homodyne measurements.

\textbf{Measurement scheme}. 
First, we introduce some notation.
By $\hat{q}_j$ and $\hat{p}_j$ we denote, respectively, the conjugated position and momentum \emph{phase-space quadrature operators} of the $j$-th mode in the canonical convention \cite{RMPReview,MyReview}, 
i.e., with the commutation relations $[\hat{q}_j,\hat{p}_{j'}]=\i\ \delta_{j,j'}$.
The \emph{particle number operator} of the $j$-th mode can be written in terms of the phase-space quadratures as $\hat{n}_j=\hat{q}^2_j+\hat{p}^2_j-1/2$.
In addition, it will be convenient to group all quadrature operators into a $2m$-component column vector $\hat{\r}$, with elements 
\begin{equation}
\label{eq:r_in_terms_p_q}
  \hat{r}_{2j-1}\coloneqq\hat{q}_j 
  \quad \text{and}\quad 
  \hat{r}_{2j}\coloneqq\hat{p}_j .
\end{equation}
As already mentioned, the action of $\hat{U}$ on mode space is given by a symplectic matrix $\S\in\mathrm{Sp}(2m,\mathbb{R})$ and a displacement vector $\vec{x}\in\mathbb{R}^{2m}$. More precisely, under a Gaussian unitary $\hat{U}$, $\hat{\r}$ transforms according to the affine linear map \cite{RMPReview}
\begin{equation}
  \label{eq:coordinatetransformdirect}
  \hat{\r} \mapsto \hat{U}^{\dagger}\hat{\r}\hat{U} =\S\hat{\r}+\vec{x}.
\end{equation}
Equivalently, the right-hand side of this equation defines the Heisenberg representation of $\hat{\r}$ with respect to $\hat{U}$. In addition, it will be useful to denote the Heisenberg representation of $\hat{\r}$ with respect to $\hat{U}^{\dagger}$ by $\hat{\tilde{\r}}\coloneqq \hat{U}\hat{\r}\hat{U}^{\dagger}$. 
Thanks to \eq~\eqref{eq:coordinatetransformdirect}, we can write $\hat{\tilde{\r}}$ in terms of the  symplectic matrix $\S$ and displacement vector $\vec{x}$ that define $\hat{U}$, as
\begin{equation}
\label{eq:coordinatetransform}
\hat{\tilde{\r}} =\S^{-1}(\hat{\r}-\vec{x}) .
\end{equation}
The symbols $\hat{r}^2\coloneqq\hat{\r}^T\hat{\r}$ and $\hat{\tilde{r}}^2\coloneqq\hat{\tilde{\r}}^T\hat{\tilde{\r}}$ will  represent, respectively,  the scalar products of $\hat{\r}$ and $\hat{\tilde{\r}}$ with themselves.
Also, we will use the same notation for the Heisenberg representations of each quadrature operator with respect to $\hat{U}^{\dagger}$, i.e., $\hat{\tilde{q}}_{j}\coloneqq\hat{U}^{\dagger}\hat{q}_{j}\hat{U}$ and $\hat{\tilde{p}}_{j}\coloneqq\hat{U}^{\dagger}\hat{p}_{j}\hat{U}$.

Next, for $\beta\in\{0, n, \n\}$, we express our fidelity bounds in the general form
\begin{equation}
\label{eq:bound_general}
 F^{(\beta)}=1-\left\langle \hat{N}^{(\beta)}\right\rangle_{\rhop},
\end{equation}
where $\hat{N}^{(\beta)}$ is an observable decomposed explicitly in terms of the local observables to which Arthur has access.
We start with the Gaussian case $\rhot \in \CG$. To express the bound \eqref{eq:boundfirst} as in \eq~\eqref{eq:bound_general}, we first write the total photon-number operator as
\begin{equation}
 \hat{n}=\sum_{j=1}^m\hat{n}_j=\sum_{j=1}^m(\hat{q}^2_j+\hat{p}^2_j-\frac{1}{2})=\hat{r}^2-\frac{m}{2}.
\end{equation}
This, in combination with \eq~\eqref{eq:boundfirst}, yields
\begin{align}
\label{eq:gen_nullifier_0}
\hat{N}^{(0)}\coloneqq \hat{\tilde{r}}^2-\frac{m}{2}.
\end{align}
 Note that, due to \eq~\eqref{eq:coordinatetransform}, each component of $\hat{\tilde{\r}}$ is a linear combination of at most $2m$ components of $\hat{\r}$. This implies that Arthur can obtain $\langle \hat{\tilde{r}}^2\rangle_{\rhop}$ by measuring at most $2m$ single-quadrature expectation values of the form $\langle \hat{r}_k\rangle_{\rhop}$ and $4m^2$ \emph{second moments} of the form $\Gamma^{(1)}_{k,k'}\coloneqq\langle \frac{1}{2}(\hat{r}_k\hat{r}_{k'}+\hat{r}_{k'}\hat{r}_{k})\rangle_{\rhop}$.
He can then classically efficiently combine them as dictated by $\S$ and $\vec{x}$ in \eq~\eqref{eq:coordinatetransform}. In Section~\ref{sec:DetailsmeasurementGaussian} of the SI, we give the details of this  measurement procedure, which we call $\MG$, and show that measuring $m\, \kappa$ second moments, instead of $4m^2$, is actually enough. Furthermore, in Section~\ref{sec:SettingsGaussian} of the SI, we show that only $m+3$ experimental settings suffice.

Now, proceeding in a similar fashion with the generic bound \eqref{eq:general_fidelity_bound}, we obtain
\begin{align}
\label{eq:gen_nullifier_n}
\hat{N}^{(\n)}\coloneqq\frac{1}{\n!}\Big(\hat{\tilde{r}}^2-\frac{m+2n}{2}\Big) \prod_{j=1}^m p_{n_j-1}\Big(\hat{\tilde{q}}^2_{j}+\hat{\tilde{p}}^2_{j}-\frac{1}{2}\Big).
\end{align}
Note that the observable in \eq~\eqref{eq:gen_nullifier_0} is contained as the special case $n=0$.
For target states in the class $\CLON$, $\hat{U}$ is assumed to be a passive Gaussian unitary.
Such unitaries preserve the area in phase space, i.e., if $\rhot \in \CLON$ it holds that $\hat{\tilde{r}}^2=\hat{r}^2$ (for details, see Section~\ref{sec:DetailsmeasuremenNonGaussian}
in the SI). Hence, using this and specialising to the case $\vec{n} = \mathbf{1}_n$, \eq~\eqref{eq:gen_nullifier_n} 
simplifies to 
\begin{equation}
\label{eq:gen_nullifier_one}
\hat{N}^{(n)}\coloneqq\Big(\hat{r}^2-\frac{m+2n}{2}\Big) \prod_{j=1}^n \Big(\hat{\tilde{q}}^2_{j}+\hat{\tilde{p}}^2_{j}-\frac{1}{2}\Big).
\end{equation}

Again by virtue of \eq~\eqref{eq:coordinatetransform}, Arthur can now obtain the expectation values 
of the observables in \eqs~\eqref{eq:gen_nullifier_n} and \eqref{eq:gen_nullifier_one}
by measuring \emph{$2j$-th moments} of the form $\Gamma^{(j)}_{k_1,l_1,\dots, k_j, l_j} \coloneqq\langle
\frac{1}{2^j}(\hat{r}_{k_1}\hat{r}_{l_1}+\hat{r}_{l_1}\hat{r}_{k_1})\cdots(\hat{r}_{k_j}\hat{r}_{l_j}+\hat{r}_{l_j}\hat{r}_{k_j})
\rangle_{\rhop}$ and then classically recombining them, which --- for constant $n$ --- he can do efficiently.
In Section~\ref{sec:DetailsmeasuremenNonGaussian} of the SI, we give the details of the measurement procedure to obtain $\Fone$, which we call $\MLON$. In particular, we show that, to obtain 
$\langle \hat{N}^{(n)}\rangle_{\rhop}$, estimating a total of $\landauO\big(m(4d^2+1)^{n}\big)$ $2j$-th moments, with $j\in[n+1]$, is enough.
Also, we list which moments are the relevant ones in terms of $\rhot \in \CLON$. Furthermore, in Section~\ref{sec:Settingsnon-Gaussian} of the SI, we show that only $\binom m n  2^{n+1}$ experimental settings suffice.

Finally, in the SI, we derive a bound analogous to that of \eqs~\eqref{eq:bound_general} with \eqref{eq:gen_nullifier_n} for post-selected target states $\rhosyst$. 
More precisely, we show that the fidelity $\FSys\coloneqq F(\rhosyst,\rhosysp)$ between $\rhosyst$ and an arbitrary, unknown $(m-a)$-mode system preparation $\rhosysp$ is lower bounded as 
\begin{align}
\label{eq:bound_PS}
\FSys&\geq\FSysn=1-\left\langle \hat{N}^{(\n)}_\Sys\right\rangle_{\rhosysp},\\
\intertext{with}
\label{eq:gen_nullifier_PS}
\hat{N}^{(\n)}_\Sys&\coloneqq\frac{\mathbb{P}(\bphi_\A|\rhot)-1 + \frac{1}{\n!}\bra{\bphi}_\A\hat{N}^{(\n)}\ket{\bphi}_\A}{\mathbb{P}(\bphi_\A|\rhot)}.
\end{align}
From this, the corresponding expressions for the classes $\CLPSG$ an $\CLPSLON$ follow, in turn, as the two particular cases $\vec{n}=\vec{0}$ and $\vec{n} = \mathbf{1}_n$ with $\hat{U}$ passive, respectively. See Section~\ref{sec:PS_fidelity_bound} of the SI for details.

\textbf{Non-Gaussian state nullifiers}. 
It is instructive to mention that the operators 
\begin{equation}
  \label{eq:Gaussian_nullifiers}
  \hat{N}^{(0)}_j\coloneqq\hat{\tilde{q}}^2_j+\hat{\tilde{p}}^2_j-1/2,
  \end{equation}
for $j\in[m]$, correspond to the so-called \emph{nullifiers} of the Gaussian states in $\CG$. The nullifiers are commuting operators  that, despite originally introduced  \cite{Gu09} as a tool to define Gaussian graph states, can be tailored to define any pure Gaussian state \cite{Aolita10,Menicucci11}: If a state is the simultaneous null-eigenvalue eigenstate of all $m$ nullifiers of a given pure Gaussian state, then the former is necessarily equal to the latter. The bound $\Fzero$, given by \eqs~\eqref{eq:bound_general} and \eqref{eq:gen_nullifier_0}, exploits the fact that if a preparation gives a sufficiently low expectation value for the sum $\hat{N}^{(0)}=\sum_{j=1}^m\hat{N}^{(0)}_j$ of all $m$ nullifiers then its fidelity with the target state must be high. A similar intuition has been previously exploited \cite{Chen13,Yokoyama13} to experimentally check for multimode entanglement of ultra-large Gaussian cluster states. Here, we can not only certify 
entanglement but the quantum state itself.

Analogously, in the non-Gaussian case, from the derivation of \eq~\eqref{eq:gen_nullifier_n}, we can identify the operator 
\begin{equation}
  \label{eq:non_Gaussian_nullifiers}
  \hat{N}^{(\n)}_j  \coloneqq \left(\hat{\tilde{q}}^2_j+\hat{\tilde{p}}^2_j-\frac{1+2n_j}{2}\right)\prod_{k=1}^m p_{n_k-1}\left(\hat{\tilde{q}}^2_{k}+\hat{\tilde{p}}^2_{k}-1/2\right)
\end{equation}
 as the $j$-th nullifier of the $m$-mode non-Gaussian state $\rhot$ of \eq~\eqref{eq:general_target}. Indeed, all $m$ observables given by \eq~\eqref{eq:non_Gaussian_nullifiers} for all $j\in[m]$ commute and have $\rhot$ as their unique, simultaneous null-eigenvalue eigenstate. 
 To end up with,  due to the projection onto $\ket{\bphi}_\A$, the equivalent observables for post-selected target states do not in general commute. Nevertheless, their sum, given by $\hat{N}^{(\n)}_\Sys$, still defines an observable with $\rhosyst$ as its unique null-eigenvalue eigenstate. 
These observables constitute, to our knowledge \cite{Aolita10,Menicucci11,RMPReview}, the first examples of nullifiers for non-Gaussian states.

\section*{Acknowledgements}
We would like to thank F.\ G.\ S.\ L.\ Brand\~ao and S.\ T.\ Flammia for discussions on certification of state preparation. We thank the EU (RAQUEL, SIQS, AQuS, REQS - Marie Curie IEF No 299141), the BMBF, the FQXI, the Studienstiftung des Deutschen Volkes, MPQ-ICFO, and FOQUS for support.

\section*{Author Contributions}
All four authors participated in all the discussions and contributed with insights.
LA conceived the fidelity bound and its estimation technique.
LA, CG, and MK carried out all the calculations and worked out the details of the formalism.


\vfill 
\pagebreak 
\onecolumngrid 
\renewenvironment{widetext}{}{} 

\renewcommand\thesection{S\arabic{section}}
\renewcommand\theequation{S\arabic{equation}}
\renewcommand\thestandardbox{S\arabic{standardbox}}
\renewcommand\thetheorem{S\arabic{theorem}}

\setcounter{equation}{0}
\setcounter{section}{0}
\setcounter{theorem}{0}
\setcounter{standardbox}{0}

\section*{\normalfont\noindent\huge\sffamily Supplementary Information}

In this Supplementary Information we present the technical details of the certification test and the proofs of the theorems.
It is organised as follows:
In Section~\ref{sec:measurement_scheme} we provide a detailed description of the measurements schemes $\MG$ and $\MLON$ for the classes of target states $\CG$ and $\CLON$, respectively. In particular, in Boxes~\ref{box:measurement_scheme} and \ref{box:measurement_scheme_1n} in that section, a full specification of the necessary correlators to measure is given.
In Section~\ref{sec:post_selected_Targets}, we extend Theorems~\ref{thm:certification_gaussian} and \ref{thm:certification_1n} to the classes of post-selected target states $\CLPSG$ and $\CLPSLON$.
Section~\ref{Proofs} contains the proofs of our theorems, as well as of the Corollaries for post-selected target states. In particular, Lemmas~\ref{lem:fidelity_bound_estimation} and \ref{lem:fidelity_bound_estimation_1n} provide more precise expressions of the bounds \eqref{eq:numb_copies_bound_Gaussian} and \eqref{eq:numb_copies_bound} of the Theorems~\ref{thm:certification_gaussian} and \ref{thm:certification_1n} in the main text. In Section~\ref{sec:meas_settings} we upper-bound the number of  experimental settings necessary for our measurement schemes. In Section~\ref{sec:systematic_errors} we analyse the stability of our fidelity estimates under systematic errors. Finally, Section~\ref{sec:auxiliary_maths} contains some auxiliary mathematical relations necessary for our treatments. 
Equation and theorem numbers that do not start with an upper case S refer to the respective equations and theorems of the main text.



\titleformat{\section}[display]{\vspace{-1em}}{}{0pt}{\normalfont\bf\sffamily}[\vspace{-1em}]


\titleformat{\section}[display]{\vspace{-1em}}{}{0pt}{\normalfont\bf\sffamily\thesection\quad}[\vspace{-1em}]
\titleformat{\subsection}[display]{\vspace{-1em}}{}{0pt}{\normalfont\bf\sffamily\thesection.\thesubsection\quad}[\vspace{-1em}]

\section{The measurement scheme}
\label{sec:measurement_scheme}
In this section we elaborate on the fidelity bounds $\Fzero$ and $\Fone$ of the fidelity bounds for the Gaussian and linear-optical case, respectively.
To this end, it will be convenient to first specify some details of the symplectic matrix $\S$, which describes the optical network.

By virtue of the Euler decomposition \cite{Braunstein,RMPReview}, $\S$ can be decomposed as
\begin{equation}
  \label{eq:Euler_decomp}
  \S=\O\,\D\, \O',
\end{equation}
where $\D\in\mathbb{R}^{2m\times 2m}$ is positive-definite and diagonal, with elements $D_{2j-1,2j-1}\coloneqq s_j\geq1$ and $D_{2j,2j}\coloneqq s^{-1}_j$, for $j\in[m]$,
and $\mathbf{O}\in\mathbb{R}^{2m\times 2m}$ and $\mathbf{O}'\in\mathbb{R}^{2m\times 2m}$ are orthogonal matrices. 
$\D$ describes $m$ active single-mode squeezers in parallel, each one with squeezing parameter $s_j$ along the position quadrature. 
The maximum single-mode squeezing is $\smax\coloneqq\max_{1\leq j\leq m}\{s_j\}$.
$\O$ and $\O'$, in turn, describe passive mode transformations that can be implemented by linear-optical networks of at most $m(m-1)/2$ beam-splitters and single-mode phase shifters \cite{Reck}.
In the two settings considered here, i.e., for any $\rhot\in\CG \cup \CLON$, the unitary $\hat{U}$ in \eqs~\eqref{eq:SG} and \eqref{eq:SLON}
is such that $\O'$ can be taken as the identity matrix. In the first setting, i.e., for $\rhot\in\CG$, this holds because $\hat{U}$ acts on the vacuum state vector $\ket{\0}$ and any passive mode transformation maps the vacuum into itself.
For the second setting, i.e., for $\rhot\in\CLON$,  this holds simply because there we assume  that the total transformation itself is passive, i.e., in that case it holds also that $\D=\1$, so that $\S=\O$. 

In both cases, coupling between different modes only takes place through the linear-optical network described by $\O$. 
A general circuit can couple all $m$ modes with each other, meaning that the quadrature operators of each output mode are linear combinations of those of all $m$ input modes. 
However, often, each mode is only coupled to at most $d\leq m$ other modes. 
In these situations, $\O$ is a sparse matrix with at most $4 m d$ non-zero elements. 
More precisely, the columns of $\O$ are given by $2m$ orthonormal vectors $(\o^{(k)})_{k\in [2m]}$ each having at most $2d$ non-zero entries. 
Furthermore, since the position and momentum of each mode is coupled to at most the $2d$ quadratures of the same $d$ modes, each pair $\o^{(2j-1)}$ and $\o^{(2j)}$ shares the same \emph{sparsity property}, i.e., $\o^{(2j-1)}$ and $\o^{(2j)}$ have at least $2(m-d)$ zero entries in common, for all $j\in [m]$. 

\subsection{Gaussian case}
\label{sec:DetailsmeasurementGaussian}

Using that in the Gaussian case $\S=\O\, \D$ and squaring \eq~\eqref{eq:coordinatetransform} 
yields
\begin{align}
\label{eq:rtransfromed}
\nonumber
\hat{\tilde{r}}^2
&=
\hat{\r}^{T}\O\D^{-2}\O^{-1}\hat{\r}-2\vec{x}^{T}\O\D^{-2}\O^{-1}\hat{\r}
+
\vec{x}^{T}\O\D^{-2}\O^{-1}\vec{x}
\\
&=
\Tr\left[\O\D^{-2}\O^T[\hat{\r}\hat{\r}^{T}-(2\hat{\r}-\vec{x})\vec{x}^{T}]\right],
\end{align}
where $\O^{-1}=\O^{T}$ has been used and the trace is taken not over the Hilbert space but over the $2m\times 2m$ matrix with operators as entries. 
Combining \eqs~\eqref{eq:rtransfromed}, \eqref{eq:bound_general}, and \eqref{eq:gen_nullifier_0} yields
\begin{equation}
\label{eq:previoustoboundexplicit}
\begin{split}
 \Fzero
 =1-\Tr\left[\O{\D^{-2}}\O^T[\langle \hat{\r}\hat{\r}^{T}\rangle_{\rhop}
  - (2\langle \hat\r \rangle_{\rhop} -\vec{x})\vec{x}^T]\right]
 +\frac{m}{2} .
\end{split}
\end{equation}
Now we introduce the \emph{first moment vector} $\bfgamma\in\mathbb{R}^{2m}$ and the symmetric \emph{second moment matrix} $\bfGamma^{(1)} \in \RR^{2m \times 2m}$ of $\rhop$, with components 
\begin{equation}
\label{eq:Gamma_def}
\gamma_l \coloneqq \langle \hat{r}_l \rangle_{\rhop} 
\quad \text{and} \quad 
\Gamma^{(1)}_{l,l'}\coloneqq
\biggl\langle\frac{\hat{r}_l\hat{r}_{l'}+\hat{r}_{l'}\hat{r}_{l}}{2}\biggr\rangle_{\rhop}  ,
\end{equation}
respectively.
Since the matrix $\O{\D^{-2}}\O^T$ is symmetric, it holds that
\begin{equation}
 \Tr\left[\O{\D^{-2}}\O^T[\langle \hat{\r}\hat{\r}^{T}\rangle_{\rhop} \right]
 =
 \Tr\left[\O{\D^{-2}}\O^T[\langle \hat{\r}\hat{\r}^{T}\rangle_{\rhop}^T \right],
\end{equation}
so that we can rewrite \eq~\eqref{eq:previoustoboundexplicit} in terms of the observables which Arthur has access to as
\begin{equation}
\label{eq:boundexplicit}
\begin{split}
 \Fzero
 =1-\Tr\left[\O{\D^{-2}}\O^{-1}[\bfGamma^{(1)}-(2\bfgamma-\vec{x})\vec{x}^T]\right]
 +\frac{m}{2}. 
\end{split}
\end{equation}

We will show later (see Lemma~\ref{lem:sparsity} in Section~\ref{sec:ProofsGaussian} and the discussion immediately after its proof) that the bound \eqref{eq:boundexplicit} actually depends on at most $2 m \kappa$ out of the $4m^2$ entries of $\bfGamma^{(1)}$, with $\kappa=2\min\{d^2,m\}$, as defined in \eq~\eqref{eq:kappa_def}. Thus, only the $2 m \kappa$ corresponding observables, and the $2m$ observables necessary for $\bfgamma$, as indicated in Box~\ref{box:measurement_scheme}, need to be measured.
All these observables  can be measured by homodyne detection \cite{RMPReview}. 
Furthermore, in Section~\ref{sec:SettingsGaussian} we show that only $m+3$ different measurement settings are required.
Finally, by classical post-processing, Arthur recombines his estimates according to the third step of Box~\ref{box:measurement_scheme} and obtains the fidelity estimate $\Festzero$. This last step is also efficient in $m$.

\tikzbox{
\begin{standardbox}[Measurement scheme $\MG$]
\label{box:measurement_scheme}\hfill
\begin{compactenum}[1)] 
\item 
For each $1\leq l\leq 2m$ Arthur uses $\ns_1$ copies of $\rhop$, with $\ns_1$ given by \eq~\eqref{eq:bound_for_cs1}, to measure the observable $\hat{r}_l$, obtaining an estimate $\gamma^\ast_{l}$ of the expectation value $\gamma_{l}=\langle\hat{r}_l\rangle_{\rhop}$.
\item 
For each $1\leq l\leq l'\leq 2m$ for which $(\O{\D^{-2}}\O^{-1})_{l,l'}=\sum_{k=1}^{2m}o^{(k)}_lD^{-2}_{k,k}o^{(k)}_{l'}\neq0$, Arthur uses $\ns_2$  copies of $\rhop$, with $\ns_2$ given by \eq~\eqref{eq:bound_for_cs2}, to measure the observable $\frac{1}{2}(\hat{r}_l\hat{r}_{l'}+\hat{r}_{l'}\hat{r}_{l})$, obtaining an estimate $\Gamma^{{(1)}\ast}_{l,l'}$ of the expectation values $\Gamma^{(1)}_{l,l'}=\Gamma^{(1)}_{l',l}$ in \eq~\eqref{eq:Gamma_def}.  
\item 
He obtains the estimate $\Festzero$ of $\Fzero$ by replacing in \eq~\eqref{eq:boundexplicit} the actual expectation values $\bfGamma^{(1)}$ and $\bfgamma$ by the estimates $\Gamma^{{(1)}\ast}$ and $\bfgamma^\ast$, respectively. 
\end{compactenum}
\end{standardbox}
}

\subsection{Linear-optical case}
\label{sec:DetailsmeasuremenNonGaussian}
For $\rhot \in \CLON$ the unitary $\hat U$ is assumed to be passive. Hence, one has $\vec{x} = \vec{0}$ and $\S = \O$, and it follows that 
\begin{align}
\label{eq:area_preservation}
  \hat{\tilde{r}}^2=\hat{r}^2 \, .
\end{align}
The components of $\tilde r$ are
\begin{align}\label{eq:pqHeisenberg}
  \hat{\tilde{q}}_{j}={\o^{(2j-1)}}^T\hat{\r}
  \quad \text{and} \quad 
  \hat{\tilde{p}}_{j}={\o^{(2j)}}^T\hat{\r},
\end{align}
where $\o^{(k)}$ denotes the $k$-th column of $\O$.
Defining
\begin{equation}
\label{Projector}
\P^{(j)}\coloneqq\o^{(2j-1)}{\o^{(2j-1)}}^T+\o^{(2j)}{\o^{(2j)}}^T 
\end{equation}
as the projector onto the subspace spanned by the two vectors 
$\o^{(2j-1)}$ and $\o^{(2j)}$ and using \eqs~\eqref{eq:pqHeisenberg},~\eqref{eq:bound_general}, and \eqref{eq:gen_nullifier_one}, we obtain
\begin{equation}
\label{eq:boundexplicitnonGaussian}
\Fone=1-\bigg\langle\Big(\hat{r}^2-\frac{m+2n}{2}\Big)\prod_{j=1}^{n}
\Big(\hat{\r}^{T}\P^{(j)}\hat{\r}-\frac{1}{2}\Big)\bigg\rangle_{\rhop} . 
\end{equation}

 Next, we consider the $\binom{n}{j}$ subsets of $\{1,2, \hdots , n\}$ of length $j$ and define $\Omega^{(j)}_\mu$ as the $\mu$-th of these subsets for some arbitrary ordering. With this, we expand the product inside \eq~\eqref{eq:boundexplicitnonGaussian} as 
\begin{equation}
\label{eq:productexpansion}
\prod_{j=1}^{n}\Big(\hat{\r}^{T}\P^{(j)}\hat{\r}-\frac{1}{2}\Big)=\sum_{j=0}^{n}(-1/2)^{n-j}\sum_{\mu=1}^{\binom{n}{j}}\bigotimes_{i\in\Omega^{(j)}_\mu}\hat{\r}^{T}\P^{(i)}\hat{\r}.
\end{equation}
Using that a product of traces can be written as a trace over tensor products, \eq~\eqref{eq:boundexplicitnonGaussian} can be written as
\begin{equation} \label{eq:Fone_in_terms_rr_expectations_first}
\Fone
= 
1-\bigg\langle\big(\hat{r}^2-\frac{m+2n}{2}\big)\sum_{j=0}^{n}(-1/2)^{n-j}
\sum_{\mu=1}^{\binom{n}{j}}\Tr\big[(\bigotimes_{i\in\Omega^{(j)}_\mu}\P^{(i)})(\hat{\r}\hat{\r}^{T})^{\otimes j}\big]\bigg\rangle_{\rhop}
 ,
\end{equation}
where $\bigotimes_{i \in \Omega^{(0)}_\mu=\emptyset} \P^{(i)} \coloneqq 1$ 
and the traces are again taken not over the Hilbert space but over tensors that have operators as components. 
For each $j\in [n+1]$, we introduce the \emph{$2j$-th moment tensors} 
$\bfGamma^{(j)} \in (\RR^{2m\times 2m})^{\otimes j}$ 
with components
\begin{equation}
\label{eq:def_Gamma_i}
\Gamma^{(j)}_{k_1,l_1,\dots, k_j, l_j} \coloneqq 
\biggl\langle
\frac{\hat{r}_{k_1}\hat{r}_{l_1}+\hat{r}_{l_1}\hat{r}_{k_1}}{2}\cdots \frac{\hat{r}_{k_j}\hat{r}_{l_j}+\hat{r}_{l_j}\hat{r}_{k_j}}{2}
\biggr\rangle_{\rhop}
\end{equation}
and define $\bfGamma^{(0)}\coloneqq 1$.
Clearly, these tensors are invariant under the partial transposition with respect to any $j'$-th pair of subindices $k_{j'}$ and $l_{j'}$,
\begin{equation}
\label{eq:symmetry_property}
\Gamma^{(j)}_{k_1,l_1,\dots, k_{j'}, l_{j'},\dots, k_j, l_j} 
= \Gamma^{(j)}_{k_1,l_1,\dots, l_{j'}, k_{j'},\dots, k_j, l_j}.
\end{equation}
With the definition \eqref{eq:def_Gamma_i} and the fact that each projector $\P^{(i)}$ is a symmetric matrix, \eq~\eqref{eq:Fone_in_terms_rr_expectations_first} finally becomes
\begin{widetext}
\begin{align} 
\Fone&= 
1-\sum_{j=0}^{n}(-1/2)^{n-j} \sum_{\mu=1}^{\binom{n}{j}} 
\Biggr\{
\Tr\biggl[ \Bigl(\1 \otimes \bigotimes_{i\in \Omega^{(j)}_\mu} \P^{(i)}\Bigr) 
\bfGamma^{(j+1)}
\biggr]
- \frac{m+2n}{2}\Tr\biggl[ \Bigl(\bigotimes_{i \in \Omega^{(j)}_\mu} \P^{(i)}\Bigr) 
\bfGamma^{(j)} \rangle_{\rhop}\biggr]
\Biggl\} .
\label{eq:boundexplicitnonGaussian2}
\end{align}
\end{widetext}
\tikzbox{
\begin{standardbox}[Measurement scheme $\MLON$]
\label{box:measurement_scheme_1n}
\noindent\begin{compactenum}[1)] 
\item  For each $1\leq j\leq n$, each $1\leq\mu\leq\binom{n}{j}$, and each $1\leq k_1, l_1,k_2,l_2,\hdots, k_j, l_j\leq 2m$, for which 
\begin{equation}
\Bigl(\bigotimes_{i \in \Omega^{(j)}_\mu} \P^{(i)}\Bigr)_{k_1,l_1,k_2,l_2,\hdots , k_j, l_j}\neq0,
\end{equation} 
\noindent
\item 
Arthur uses $\ns_{\leq2(n+1)}$ copies of $\rhop$, with $\ns_{\leq2(n+1)}$ given by \eq~\eqref{eq:ns_1n}, to measure the observable 
$({\hat{r}_{k_1}\hat{r}_{l_1}+\hat{r}_{l_1}\hat{r}_{k_1}})/{2}\cdots ({\hat{r}_{k_j}\hat{r}_{l_j}+\hat{r}_{l_j}\hat{r}_{k_j}})/{2}$, obtaining an estimate 
${\Gamma^{(j)\ast}}_{k_1,l_1,k_2,l_2,\hdots , k_j, l_j}$ of the $2j$-th moment $\Gamma^{(j)}_{k_1,l_1,k_2,l_2,\hdots,  k_j, l_j}$. 
For each $1\leq k_{j+1}\leq 2m$, he uses $\ns_{\leq2(n+1)}$ copies of $\rhop$ to measure the observable $(({\hat{r}_{k_1}\hat{r}_{l_1}+\hat{r}_{l_1}\hat{r}_{k_1}})/{2})\cdots
(({\hat{r}_{k_j}\hat{r}_{l_j}+\hat{r}_{l_j}\hat{r}_{k_j}})/{2})\hat{r}^2_{k_{j+1}}$, 
obtaining an estimate ${\Gamma^{(j+1)\ast}}_{k_1,l_1,k_2,l_2,\hdots  ,  k_j, l_j,k_{j+1},k_{j+1}}$ of the $2(j+1)$-th moment $\Gamma^{(j+1)}_{k_1,l_1,k_2,l_2,\hdots ,  k_j, l_j,k_{j+1},k_{j+1}}$.
\item He obtains the estimate $\Fest$ of $\Fone$ by replacing in \eq~\eqref{eq:boundexplicitnonGaussian2} for all $1\leq j\leq n+1$ the actual expectation values $\bfGamma^{(j)}$ by the estimates ${\bfGamma^{(j)\ast}}$.
\end{compactenum} 
\end{standardbox}
}
\noindent Note that this is an explicit expression for $\Fone$ in terms of the correlators \eqref{eq:def_Gamma_i} that Arthur can measure.
Due to the sparsity of $\O$, each matrix $\P^{(i)}$ has at most $(2d)^2$ 
non-zero entries. Then, it follows (see Lemma~\ref{lem:sparsity_1n} in Section~\ref{sec:ProofsNonGaussian} for details) that the measurement of 
$\landauO\left(m\left(4d^2+1\right)^n\right)$
observables, those listed in Box~\ref{box:measurement_scheme_1n}, suffices for the estimation of \eqref{eq:boundexplicitnonGaussian2}. 
As in the Gaussian case, all these observables can be measured by homodyne detection \cite{RMPReview}. Furthermore, in Section~\ref{sec:Settingsnon-Gaussian} we show that at most $\binom m n  2^{n+1}\leq(2m)^{n}/n!$ measurement settings are sufficient. Once again, by classical post-processing, Arthur recombines his estimates according to the third step of Box~\ref{box:measurement_scheme_1n} and obtains the fidelity estimate $\Festone$. Provided that $n$ is constant, this last step is also efficient in $m$. 
\section{Quantum certification of locally post-selected target states}
\label{sec:post_selected_Targets}
In this section, we extend our results to locally post-selected $(m-a)$-mode target states $\rhosyst$ in $\CLPSG$ or $\CLPSLON$. The entire treatment of the classes $\CLPSG$ or $\CLPSLON$ is similar to, and follows directly from, that already seen for the classes $\CG$ or $\CLON$. Therefore, instead of repeating all the details, we simply explain the specific differences. 
\subsection{The fidelity bound}
\label{sec:PS_fidelity_bound}
The first step is to derive 
the fidelity bound $\FSysn$ given by \eq~\eqref{eq:bound_PS}.
We proceed in a similar fashion to the Methods Section in the main text. Due to \eqs~\eqref{eq:fidelity} and \eqref{eq:general_PS_state}, the facts that $\rhosyst$ and $\rhot$ are pure, and the  properties of the trace, it holds that 
\begin{align}
\label{eq:fid_syst}
\FSys = F(\rhosyst,\rhosysp)=
\Tr_\Sys\left[\Tr_\A\left[\frac{\rhot (\1_\Sys\otimes\ket{\bphi}_\A\bra{\bphi}_\A)}{\mathbb{P}(\bphi_\A|\rhot)}\right]\rhosysp\right]
=\frac{\Tr\left[\rhot(\rhosysp\otimes\ket{\bphi}_\A\bra{\bphi}_\A)\right]}{\mathbb{P}(\bphi_\A|\rhot)}
=\frac{F\left(\rhot,\rhosysp\otimes\ket{\bphi}_\A\bra{\bphi}_\A\right)}{\mathbb{P}(\bphi_\A|\rhot)},
\end{align}
 where $\Tr_\Sys$ indicates partial trace over the Fock space of the $m-a$ modes in $\Sys$. 
Now, due to \eq~\eqref{eq:bound_general}, it holds that  
\begin{align}
F\left(\rhot,\rhosysp\otimes\ket{\bphi}_\A\bra{\bphi}_\A\right)
&\geq1 - \Tr\left[\hat{N}^{(\n)}(\rhosysp\otimes\ket{\bphi}_\A\bra{\bphi}_\A)\right]
=
1 - \Tr_\Sys\left[\bra{\bphi}_\A\hat{N}^{(\n)}\ket{\bphi}_\A\, \rhosysp\right],
\label{eq:fid_bound}
\end{align}
with $\hat{N}^{(\n)}$ the observable of \eq~\eqref{eq:gen_nullifier_n}. Using \eqs~\eqref{eq:fid_syst} and \eqref{eq:fid_bound}, we obtain the general fidelity bound $\FSysn$ of \eq~\eqref{eq:bound_PS}.

In particular, setting $\vec{n}=\vec{0}$ in \eqs~\eqref{eq:bound_PS} and~\eqref{eq:gen_nullifier_PS} yields the specialized fidelity bound $\FSyszero \geq 1- \left\langle\hat N^{(0)}_\Sys \right\rangle_{\rhosysp}$ for the case $\rhosyst \in \CLPSG$, with
\begin{align}
\label{eq:fid_bound_final_LPSG}
\hat{N}^{(0)}_\Sys\coloneqq\frac{\mathbb{P}(\bphi_\A|\rhot)-1 + \bra{\bphi}_\A\hat{N}^{(0)}\ket{\bphi}_\A}{\mathbb{P}(\bphi_\A|\rhot)} ,
\end{align}
where $\hat{N}^{(0)}$ is the observable of \eq~\eqref{eq:gen_nullifier_0} and $\rhot$ is the $m$-mode state in $\CG$ associated with $\rhosyst$ through \eq~\eqref{eq:general_PS_state}. 
Analogously, taking $\vec{n} = \mathbf{1}_n$ and $\hat{U}$ passive yields the corresponding fidelity bound $\FSysone \geq 1- \left\langle\hat N^{(n)}_\Sys \right\rangle_{\rhosysp}$ for $\rhosyst \in \CLPSLON$, with
\begin{align}
\label{eq:fid_bound_final_LPSLON}
\hat{N}^{(n)}_\Sys\coloneqq\frac{\mathbb{P}(\bphi_\A|\rhot)-1 + \bra{\bphi}_\A\hat{N}^{(n)}\ket{\bphi}_\A}{\mathbb{P}(\bphi_\A|\rhot)},
\end{align}
where $\hat{N}^{(n)}$ is the observable from \eq~\eqref{eq:gen_nullifier_n} and $\rhot$ is the $m$-mode state in $\CLON$ associated to $\rhosyst$ through \eq~\eqref{eq:general_PS_state}.

\subsection{The certification test}
\label{sec:PS_cert_test}
Next, in Box~\ref{box:test_PS}, we present a test $\mathcal{T}_\text{LPS}$ that works for post-selected target states in $\CLPSG$ or $\CLPSLON$ and which is a slightly modified version of the test $\mathcal{T}$ from Box~\ref{box:test}. It is, of course, possible to unify both tests so as to account for all four classes of target states in one single test. We have however opted for splitting the tests into the two cases with and without post-selection to avoid an excessive notational overhead in Box~\ref{box:test} of the main text.
\tikzbox{
  \begin{standardbox}[Certification test $\mathcal{T}_\text{LPS}$]
  \label{box:test_PS}
  \noindent\begin{compactenum}[1)] 
  \item Idem as in $\mathcal{T}$ from Box~\ref{box:test}. 
  \item Arthur provides Merlin with the classical specification $n$, $\S$, $\vec{x}$, $a$, and $\ket{\bphi}_\A$ of the target state $\rhosyst$ and requests a sufficient number of copies of it. 
  \item If $n=0$,
  Arthur measures $2 m \kappa$ two-mode correlations and $2 (m-a)$ single-mode expectation values specified by the measurement scheme $\MLPSG$ (see Section~\ref{sec:PS_measurement_scheme}), which can be done with $m-a+3$ single-mode homodyne settings.
  \\
  If $n>0$,
  he measures $\landauO\big(m(4d^2+1)^{n}\big)$  multi-body correlators, each one involving  between $1$ and $2n+1$ modes, specified  by the measurement scheme  $\MLPSLON$  (see Section~\ref{sec:PS_measurement_scheme}), which can be done with at most $\binom m n  2^{n+1}$ single-mode homodyne settings. 
  \item \label{itme:1.4_SI}
  By classical post-processing, he obtains a fidelity estimate $\FSysestone$ such that $\FSysestone\in [\FSysone-\varepsilon,\FSysone+\varepsilon]$ with probability at least $1-\alpha$, where $\FSysone$ is the lower bound to $\FSys$ given by \eq~\eqref{eq:fid_bound_final_LPSLON}.
  \item \label{item:accept_reject_condition_SI}
  If $\FSysestone<\FT+\varepsilon$, he rejects. Otherwise, he accepts.
  \end{compactenum} 
  \end{standardbox}
}

\subsection{The measurement scheme}
\label{sec:PS_measurement_scheme}
The measurement schemes $\MLPSG$ and $\MLPSLON$
to estimate $\FSyszero$ and $\FSysone$, respectively, are essentially replicas of the schemes $\MG$ and $\MLON$ to estimate $\Fzero$ and $\Fone$, already described in detail in boxes~\ref{box:measurement_scheme} and \ref{box:measurement_scheme_1n}. Thus, instead of repeating all the details of boxes~\ref{box:measurement_scheme} and \ref{box:measurement_scheme_1n}, we simply outline the concrete differences between $\MLPSG$ and $\MG$, as well as between $\MLON$ and $\MLPSLON$. There are only three specific differences. 
\begin{enumerate}
\item The moment vector and tensors are now defined with respect to $\rhosysp\otimes\ket{\bphi}_\A\bra{\bphi}_\A$ instead of $\rhop$. More precisely, we now need to estimate the vector $\bfgamma_\Sys\in\mathbb{R}^{2m}$ and tensors $\bfGamma_\Sys^{(j)} \in \left(\RR^{2m\times 2m}\right)^{\otimes j}$, with elements
\begin{align}
\label{eq:def_gamma_S}
{\gamma_\Sys}_l \coloneqq& \left\langle \hat{r}_l \right\rangle_{\rhosysp\otimes\ket{\bphi}_\A\bra{\bphi}_\A}
=
\left\langle \bra{\bphi}_\A\hat{r}_l\ket{\bphi}_\A \right\rangle_{\rhosysp}
=
\left\{
  \begin{array}{lr}
    \bra{\phi_l}_{\A_l}\hat{r}_l\ket{\phi_l}_{\A_l} &, \text{ if } l\in\A,\\
    \left\langle \hat{r}_l \right\rangle_{\rhosysp} &, \text{ if } l\notin\A,
  \end{array}
\right.\\
\intertext{and} 
\nonumber
{\Gamma^{(j)}_\Sys}_{k_1,l_1,\dots, k_j, l_j} \coloneqq& \left\langle
\frac{\hat{r}_{k_1}\hat{r}_{l_1}+\hat{r}_{l_1}\hat{r}_{k_1}}{2}\cdots \frac{\hat{r}_{k_j}\hat{r}_{l_j}+\hat{r}_{l_j}\hat{r}_{k_j}}{2}
\right\rangle_{\rhosysp\otimes\ket{\bphi}_\A\bra{\bphi}_\A}\\
\label{eq:def_Gamma_S}
=&
\left\langle\bra{\bphi}_\A
\frac{\hat{r}_{k_1}\hat{r}_{l_1}+\hat{r}_{l_1}\hat{r}_{k_1}}{2}\cdots \frac{\hat{r}_{k_j}\hat{r}_{l_j}+\hat{r}_{l_j}\hat{r}_{k_j}}{2}
\ket{\bphi}_\A\right\rangle_{\rhosysp},
\end{align}
respectively. 
\item $\FSyszero$ and $\FSysone$ are obtained dividing the expressions on the right-hand sides of \eqs~\eqref{eq:boundexplicit} and \eqref{eq:boundexplicitnonGaussian2} by $\mathbb{P}(\bphi_\A|\rhot)$, and with $\bfgamma$ and  $\bfGamma^{(j)}$ replaced by $\bfgamma_\Sys$ and $\bfGamma_\Sys^{(j)}$, respectively. 
\item The presence of the divisor $\mathbb{P}(\bphi_\A|\rhot)$ in $\FSyszero$ and $\FSysone$ is the reason for the third difference. 
As discussed in Lemmas \ref{lem:stability_PS} and \ref{lem:stability_1n_PS} in Sections~\ref{sec:proof_corollary_G} and \ref{sec:proof_corollary_LON}, respectively, this divisor makes $\FSyszero$ and $\FSysone$ $1/\mathbb{P}(\bphi_\A|\rhot)$ times more unstable than $\Fzero$ and $\Fone$
. As a consequence, the number of copies of $\rhosysp$ required to estimate each relevant moment of $\FSyszero$ are $\frac{\ns_1}{\mathbb{P}(\bphi_\A|\rhot)}$ and $\frac{\ns_2}{\mathbb{P}(\bphi_\A|\rhot)}$, instead of $\ns_1$ and $\ns_2$. This is summarized in Lemma~\ref{lem:fidelity_bound_estimation_PS}. Analogously, the number required for each relevant moment of $\FSysone$ is $\frac{\ns_{\leq2(n+1)}}{\mathbb{P}(\bphi_\A|\rhot)}$, instead of $\ns_{\leq2(n+1)}$. This is summarized in Lemma~\ref{lem:fidelity_bound_estimation_1n_PS}.
\end{enumerate}

As is clear from \eqs~\eqref{eq:def_gamma_S} and \eqref{eq:def_Gamma_S}, the estimation of $\bfgamma_\Sys$ and $\bfGamma_\Sys^{(j)}$  requires only the measurement of multi-body correlators  among  the $(m-a)$ system modes in $\Sys$. 
This is due to the facts that after post selection the system is in a product state with respect to the bipartition in $\Sys$ and $\A$ and that the quadrature operators in \eq~\eqref{eq:def_Gamma_S} can also be correspondingly grouped into two factors, one containing exclusively operators of modes in $\Sys$ and the other in $\A$. 
Furthermore, since $\ket{\bphi}_\A$ is a product state known to Arthur, he can efficiently calculate the expectation vale of any product of quadrature operators of modes in $\A$ with respect to $\ket{\bphi}_\A$. 
For instance, suppose that $k_1, l_1, k_2\in\A$ and that $l_2,k_3, l_3 \dots, k_j, l_j\notin\A$. Then, the corresponding $2j$-th moment decomposes as
\begin{align}
\label{eq:exp_value_factor}
{\Gamma^{(j)}_\Sys}_{k_1,l_1,\dots, k_j, l_j}=
\bra{\bphi}_\A
\frac{\hat{r}_{k_1}\hat{r}_{l_1}+\hat{r}_{l_1}\hat{r}_{k_1}}{2}\hat{r}_{k_2}\ket{\bphi}_\A
\left\langle\hat{r}_{l_2}\frac{\hat{r}_{k_3}\hat{r}_{l_3}+\hat{r}_{l_3}\hat{r}_{k_3}}{2}\cdots \frac{\hat{r}_{k_j}\hat{r}_{l_j}+\hat{r}_{l_j}\hat{r}_{k_j}}{2}\right\rangle_{\rhosysp},
\end{align}  
and only the measurement of the $(2j-3)$-th moment given by the second factor of \eq~\eqref{eq:exp_value_factor} is required. As another example, consider the case where a given $\ket{\phi_j}_{\A_j}$ is a Fock-basis state. Then, all the moments containing an odd number of quadrature operators of the $\A_j$-th mode automatically vanish and need therefore not be measured at all. 

In  general, Arthur can always efficiently obtain $\bfgamma_\Sys$ and the $\bfGamma_\Sys^{(j)}$'s as a product of an (a priori known) expectation value with respect to $\ket{\bphi}_\A$ of a multi-body product of quadrature operators of modes in $\A$ and a (measured) expectation value with respect to $\rhosysp$ of a multi-body product of quadrature operators of modes in $\Sys$, in a way analogous to the example of \eq~\eqref{eq:exp_value_factor}. 

\subsection{Corollaries of Theorems \ref{thm:certification_gaussian} and \ref{thm:certification_1n}}
\label{sec:PS_Corollary}
Since the moments to be estimated are now given, in \eqs~\eqref{eq:def_gamma_S} and \eqref{eq:def_Gamma_S}, by expectation values with respect to $\rhosysp\otimes\ket{\bphi}_\A\bra{\bphi}_\A$, instead of $\rhop$, a simple way to extend Theorems~\ref{thm:certification_gaussian} and \ref{thm:certification_1n} to target states in  $\CLPSG$ or $\CLPSLON$ is by redefining the variance upper bounds $\sigma_i$. More precisely, taking $\sigma_i$  as an upper bound on the variances of any product of $i$ phase space quadratures now in the state $\rhosysp$, we introduce the quantities 
\begin{equation}
\label{eq:new_variance}
\varsigma_i\coloneqq\max_{j\in[a]\wedge k_1,k_2, \hdots k_j\in\A}\left\{\bra{\bphi}_\A\hat{r}_{k_1}\hat{r}_{k_2}\hdots\hat{r}_{k_j}\ket{\bphi}_\A\,\sigma_{i-j}\right\}, 
\end{equation}
for $i\in[2(n+1)]$. In addition, we call $\varsigma_{\leq i} \coloneqq \max_{k \leq i}\{\varsigma_k\}$ the \emph{maximal $i$-th generalised variance} of $\rhosysp$. 

The parameters $\varsigma_i$ quantify the maximal variances of random variables defined by products of $i-j$ quadrature-measurement outcomes on $\rhosysp$ renormalised by the expectation value of products of $j$ quadrature operators with respect to $\ket{\bphi}_\A$, therefore automatically accounting for factorisations of the type of \eq~\eqref{eq:exp_value_factor}. They constitute very non-tight upper bounds to the real variances. In particular experimental situations, tighter bounds can be found. Here, we are simply interested in taking advantage of the proofs of Theorems \ref{thm:certification_gaussian} and \ref{thm:certification_1n} without introducing too much extra notational overhead, for which the definition of \eq~\eqref{eq:new_variance} is enough.
Indeed, with these redefinitions, the following corollaries follow straightforwardly from Theorems \ref{thm:certification_gaussian} and \ref{thm:certification_1n}.
\begin{corollary}[Quantum certification of locally post-selected Gaussian states] 
\label{cor:certification_post_selected_G}
Under the same conditions and for the same $\rhot$ as in Theorem~\ref{thm:certification_gaussian}, test $\mathcal{T}_\text{LPS}$ from Box~\ref{box:test_PS} is a certification test for $\rhosyst\in\CLPSG$ and requires at most 
\begin{equation}
\label{eq:numb_copies_PS_Gaussian}
\landauO\!\left( 
 \frac{\smax^4\left(2\varsigma_1^2 \norm{\vec x}_2^2 m^3 +  \varsigma_2^2 \kappa^3 m^4\right)}
  {\left[\mathbb{P}(\bphi_\A|\rhot)\,\varepsilon\right]^2  \ln(1/(1-\alpha))} 
 \right) 
\end{equation}
copies of a preparation $\rhosysp$ with first and second generalized variance bounds $\varsigma_{1}>0$ and $\varsigma_{2}>0$, respectively.
\end{corollary}

\begin{corollary}[Quantum certification of locally post-selected linear-optical network states] 
\label{cor:certification_post_selected_LON}
Under the same conditions and for the same $\rhot$ as in Theorem~\ref{thm:certification_1n}, test $\mathcal{T}_\text{LPS}$ from Box~\ref{box:test_PS} is a certification test for $\rhosyst\in\CLPSLON$ and requires at most 
\begin{equation}
\label{eq:numb_copies_PS_LON}
\landauO\!\left(
 \frac{\varsigma_{\leq 2 (n+1)}^2 m^4 (\lambda\, d^{6}\, n\, m)^{n}}{\left[\mathbb{P}(\bphi_\A|\rhot)\,\varepsilon\right]^2 \ln(1/(1-\alpha))}
 \right)
 \end{equation}
copies of a preparation $\rhosysp$  with maximal $2(n+1)$-th generalised variance $\varsigma_{\leq2(n+1)}$, where $\lambda > 0$ is the same absolute constant as in Theorem~\ref{thm:certification_1n}.
\end{corollary}
Corollary \ref{cor:certification_post_selected_G} is proven in Section~\ref{sec:proof_corollary_G} and Corollary \ref{cor:certification_post_selected_LON} in Section~\ref{sec:proof_corollary_LON}. Equations~\eqref{eq:numb_copies_PS_Gaussian} and \eqref{eq:numb_copies_PS_LON} correspond to exactly the same expressions as in \eqs~\eqref{eq:numb_copies_bound_Gaussian} and \eqref{eq:numb_copies_bound}, respectively, with the replacements  $\sigma\to\varsigma$ and $\varepsilon\to\mathbb{P}(\bphi_\A|\rhot)\,\varepsilon$. The rescaling of $\varepsilon$ with the factor $\mathbb{P}(\bphi_\A|\rhot)$ originates directly from the new expression for the fidelity given in \eq~\eqref{eq:fid_syst}. As mentioned in Section~\ref{sec:PS_measurement_scheme}, this makes the fidelity bounds $\FSyszero$ and $\FSysone$ more unstable than $\Fzero$ and $\Fone$, leading to the error rescaling discussed earlier. 
In most interesting cases, the post-selection success probability $\mathbb{P}(\bphi_\A|\rhot)$ turns out to be 
exponentially small in $a$. 
Moreover, one can always come up with families of target states and post selection procedures for which $\mathbb{P}(\bphi_\A|\rhot)$ decreases arbitrarily fast in $m$.
In such cases, the scalings in \eqs~\eqref{eq:numb_copies_PS_Gaussian} and \eqref{eq:numb_copies_PS_LON} are not efficient in $m$, inheriting the inefficiency of the state preparation by local measurements and post selection. However, both bounds are efficient in $1/\PP(\bphi_\A|\rhot)$. That is, in every practical situation where state preparation via post selection is feasible, so is state certification. 
Interestingly, even for families of target states and post selection procedures for which $\PP(\bphi_\A|\rhot)$ decays exponentially in $a$, the overall scaling (of both bounds) with $a$ is better than the scalings (of the bounds in \eqs~\eqref{eq:numb_copies_bound} and \eqref{eq:numb_copies_PS_LON}) with $n$. 
Indeed, the bound in \eq~\eqref{eq:numb_copies_PS_LON} grows, just like the bound in \eq~\eqref{eq:numb_copies_bound}, faster than exponentially in $n$. 
Finally, both bounds \eqref{eq:numb_copies_PS_Gaussian} and \eqref{eq:numb_copies_PS_LON} scale polynomially with all the other relevant parameters, including $1/\varepsilon$. Thus, arbitrary $m$-mode target states from the classes $\CLPSG$ and $\CLPSLON$, with constant $n$, are certified by $\mathcal{T}_\text{LPS}$ efficiently in $m$, $1/\PP(\bphi_\A|\rhot)$, and all the other relevant parameters.

\subsection{Corollary of Theorem~\ref{thm:strong_certification}}
\label{sec:PS_Corollary_strong}
Finally, it is possible to show that our certification test is robust also for the locally post-selected target states of the classes $\CLPSG$ or $\CLPSLON$. 
Writing $\rhosysp$ as
\begin{equation}
\label{eq:rhosysp_in_terms_rhosys}
\rhosysp=\FSys\rhosyst+(1-\FSys)\rhosyst^{\perp},
\end{equation}
where $\rhosyst^{\perp}$ is such that $\Tr[{\rhosyst}\,\rhosyst^{\perp}]=0$,
and introducing the \emph{generalised photon mismatch} $\tilde{n}_\Sys^{\perp}$ between $\rhosyst$ and $\rhosysp$ as 
\begin{align}
\label{eq:def_phononmismatch_PS}
\ntperp_\Sys\coloneqq\left\langle \frac{\mathbb{P}(\bphi_\A|\rhot)-1 + (\hat{n}-n)\prod_{j=1}^{n}\hat{n}_{j}}{\mathbb{P}(\bphi_\A|\rhot)}
\right\rangle_{\hat{U}^{\dagger}\rhosyst^{\perp}\otimes\ket{\bphi}_\A\bra{\bphi}_\A\hat{U}}=\left\langle \hat{N}^{(n)}_\Sys\right\rangle_{\rhosyst^{\perp}},
\end{align}
where $\hat{N}^{(n)}_\Sys$ is the same observable as in \eqref{eq:fid_bound_final_LPSLON}, the following holds true.
\begin{corollary}[Robust quantum certification of locally post-selected states] 
\label{col:strong_certification_PS}
Under the same conditions as in Corollaries~\ref{cor:certification_post_selected_G} and \ref{cor:certification_post_selected_LON}, test $\mathcal{T}$ from Box~\ref{box:test} is a robust certification test for $\rhosyst \in \CLPSG \cup \CLPSLON$ with fidelity gap 
\begin{equation}
\label{eq:delta_PS}
\Delta_\Sys\coloneqq\max\left\{\frac{2\varepsilon+(1-\FT)(\ntperp_\Sys-1)}{\ntperp_\Sys},2\varepsilon\right\},
\end{equation}
where $\ntperp_\Sys$ is the generalised photon mismatch. 
\end{corollary}
The proof is identical to the proof of Theorem~\ref{thm:strong_certification} presented in Section~\ref{Proofofthm:strong_certificationGaussian} but with the replacements $F\to\FSys$, $\Fone\to\FSysone$, $\Festone\to\FSysestone$, $\Delta\to\Delta_\Sys$, and $\ntperp\to\ntperp_\Sys$.

\section{Proofs of the theorems and corollaries}
\label{Proofs}
Before going to the proofs, we devote two sections to establish necessary notation, review some known facts, and prove a general lemma.

\subsection{Norms}
Here, we introduce some helpful notation used in the proofs and review a few facts about norms on finite dimensional vector spaces. 
The \emph{max norm} $\norm{\cdot}_{\max{}}$ of a tensor is the largest of the absolute values of its entries. For a matrix $\vec A$, for example, $\norm{\vec A}_{\max{}} \coloneqq \max_{k,l} |A_{k,l}|$.
For $p \in [1, \infty]$, we denote the \emph{vector $p$-norm} of a vector $\vec a$ by $\norm{\vec a}_p$ and the \emph{Schatten $p$-norm} of a matrix $\vec A$ by $\norm{\vec A}_p$, which is defined to be the vector $p$-norm of the vector of its singular values.
For any matrix $\vec A$, we define $\vect(\vec A)$ to be a vector containing all the entries of $\vec A$ (in some order). Then one can see that
\begin{equation}\label{eq:2norm_property}
 \norm{\vec A}_2 = \norm{\vect(\vec A)}_2  
\end{equation}
and 
\begin{equation}\label{eq:max_norm_property}
 \norm{\vec A}_{\max{}} = \norm{\vect(\vec A)}_\infty  .
\end{equation}
For the vector and Schatten $p$-norm of vectors with $N$ elements and $N\times N$ matrices, respectively, the following inequalities hold
\begin{equation}\label{eq:norm_ineqs}
 \norm{\argdot}_1 \leq \sqrt N \norm{\argdot}_2 \leq N \norm{\argdot}_\infty  .
\end{equation}
Because the Schatten $\infty$-norm is induced by the vector $2$-norm, i.e.,
\begin{equation}
 \norm{\vec A}_\infty = \sup_{\vec y} \frac{\norm{\vec A \vec y}_2}{\norm{\vec y}_2}  ,
\end{equation}
it follows that for any two vectors $\bfepsilon$ and $\vec x$
\begin{equation}\label{eq:norm_of_rank_one_matrix}
 \norm{\bfepsilon \vec x^T}_\infty \leq \norm{\bfepsilon}_2 \norm{\vec x}_2  .
\end{equation}

\subsection{Reliable estimation of expectation values from samples} 
We continue by proving a general large-deviation bound for estimates of expectation values from a finite number of measurements on independent copies, which we need for the proofs of Theorems~\ref{thm:certification_gaussian} and \ref{thm:certification_1n}.

\begin{lemma}[Reliable estimation of multiple expectation values from samples]
\label{lem:general_estimation}
Let $\sigma >0$, $\rho$ be a state, and let $\hat A_1, \dots, \hat A_N$ be observables with expectation values 
$A_i \coloneqq \Tr[\rho \hat A_i]$ and variances bounded as 
$\Tr[\rho \hat A_i^2] - A_i^2 \leq \sigma^2$.
For each $i \in [N]$ and $\chi$, let $X_i^{(\chi)}$ be the random variable given by the measurement statistics of $\hat A_i$ on state $\rho$; such that, in particular, the $(X_i^{(\chi)})_{i,\chi}$ are independent random variables 
 and the \emph{finite sample average} over $\nsA$ measurements of $\hat A_i$ is the random variable
\begin{equation} \label{eq:def_sample_average}
  A_i^\ast \coloneqq \kw \nsA \sum_{\chi=1}^\nsA X_i^{(\chi)} . 
\end{equation} 
Then, the $\{A_i^\ast\}_i$ are independent and, for every $\epsilon>0$ and $\ol \alpha\in [1/2,1)$, it holds that
\begin{equation}
\label{eq:suficient}
 \PP\bigl[ \forall i : \ |A_i^\ast - A_i| \leq \varepsilon \bigr] \geq \ol \alpha   
\end{equation}
whenever
\begin{equation}
 \nsA \geq \frac{\sigma^2(N+1)}{\varepsilon^2\ln(1/\ol \alpha)} .
\end{equation}
\end{lemma}

\begin{proof}
The sample averages $\{A_i^\ast\}_i$ are independent by definition.
By Chebyshev's inequality it holds that
\begin{equation}
  \forall i\in[\nsA] :\quad \PP\bigl[ |A_i^\ast - A_i| > \varepsilon \bigr] < \frac{\sigma^2}{\nsA \varepsilon^2} . 
\end{equation}
Since the $\{A_i^\ast\}_i$ are independent random variables, this yields
\begin{equation}
 \PP\bigl[ \forall i : \ |A_i^\ast - A_i| \leq \varepsilon \bigr] 
 \geq \left(1 - \frac{\sigma^2}{\nsA \varepsilon^2}\right)^N .
\end{equation}
Finally,
\begin{equation}
 \left(1 - \frac{\sigma^2}{\nsA \varepsilon^2}\right)^N \geq \ol \alpha 
\end{equation}
is satisfied if
\begin{equation}
 \nsA \geq \nsA_{\mathrm{opt}} \coloneqq \left\lceil \frac{\sigma^2/\varepsilon^2}{1-\ol\alpha^{1/N}} \right\rceil .
\end{equation}
To finish the proof we upper bound
\begin{equation}
 \nsA_{\mathrm{opt}} = \left\lceil \frac{\sigma^2/\varepsilon^2}{1-\e^{-\frac{\ln(1/\ol \alpha)}{N}}} \right\rceil .
\end{equation}
Using that (see Section~\ref{sec:app_bound}) for all $x \geq 0$
\begin{equation} 
\label{eq:bound_for_cs}
 \frac{1}{1-\e^{-1/x}} \leq x + \frac 1 {2+2x} + \frac 1 2,
\end{equation}
it follows that
\begin{equation}
 \nsA_{\mathrm{opt}} \leq 
 \frac{\sigma^2}{\varepsilon^2} 
 \left( \frac{N}{\ln(1/\ol \alpha)} + \kw{2+\frac{2N}{\ln(1/\ol \alpha)}} +\kw 2 \right) .
\end{equation} 
To simplify the right-hand side of this inequality, we use that, since $\ol \alpha \geq \frac{1}{2}\geq\e^{-1}$, it holds that $\ln(1/\ol \alpha)\leq 1$ and, therefore, $2+\frac{2N}{\ln(1/\ol \alpha)}\geq4$. So, using again that $\ln(1/\ol \alpha)\leq 1$, we finally arrive at
\begin{equation}
 \nsA_{\mathrm{opt}} \leq 
 \frac{\sigma^2}{\varepsilon^2} 
 \left( \frac{N}{\ln(1/\ol \alpha)} + \frac{3}{4} \right) \leq \frac{\sigma^2(N+1)}{\varepsilon^2\ln(1/\ol \alpha)}.
\end{equation} 
\end{proof}

\subsection{Proof of Theorem~\ref{thm:certification_gaussian}}
\label{sec:ProofsGaussian}
Before the proof of Theorem~\ref{thm:certification_gaussian}, we present three auxiliary lemmas  specific to the fidelity bound $\Fzero$ for the Gaussian case.

The first lemma upper-bounds the number of elements of $\bfGamma^{(1)}$ which the fidelity bound $\Fzero$ depends on.
\begin{lemma}[Sparsity of the Gaussian fidelity bound]
\label{lem:sparsity}
$\Fzero$ depends on at most $2 m \kappa$ elements of $\bfGamma^{(1)}$.
We call these the \emph{relevant elements} of $\bfGamma^{(1)}$.
\end{lemma}

\begin{proof}
\Eq~\eqref{eq:boundexplicit} can be written as
\begin{equation}\label{eq:Fzero_expanded}
 \Fzero
 =1 +\frac{m}{2}
 +\vec{x}^T\O \D^{-2} \O^T(2\bfgamma-\vec{x}) 
 -\Tr\bigl[\O \D^{-2} \O^T \bfGamma\bigr]  .
\end{equation}
The last term can, in turn, be expressed as
\begin{align}
\label{eq:counting_sparsity}
\Tr\bigl[\O \D^{-2} \O^T \bfGamma^{(1)}\bigr] 
 &= 
 \sum_{k=1}^{2m} D^{-2}_{k,k} (\o^{(k)})^T\bfGamma^{(1)}\o^{(k)} 
  \nonumber \\
 &=
  \Tr\biggl[ \sum_{j=1}^m \Bigl\{
    s_j^{-2} \o^{(2j-1)} (\o^{(2j-1)})^T
    +
    s_j^2 \o^{(2j)} (\o^{(2j)})^T\Bigr\}\bfGamma^{(1)}
 \biggr]
 . 
\end{align}
Due to the sparsity of $\O$, as described in Section~\ref{sec:measurement_scheme}, each matrix $s_j^{-2} \o^{(2j-1)} (\o^{(2j-1)})^T + s_j^2 \o^{(2j)} (\o^{(2j)})^T$ has at most $4d^2$ non-zero elements. Hence, summing over $j$, we see that $\Fzero$ depends on at most $4m\min\{d^2,m\} = 2 \kappa m$ elements of $\bfGamma^{(1)}$. 
\end{proof}
Note that the counting argument following \eq~\eqref{eq:counting_sparsity} does not take into account the fact that $\bfGamma^{(1)}$ is symmetric. Taking this fact into account, we see that, from the $4d^2$ relevant elements of $\bfGamma^{(1)}$ that appear in each term of the trace \eqref{eq:counting_sparsity}, only $d(2d+1)$ are independent. Thus, even though $2 m \kappa$ entries of $\bfGamma^{(1)}$ contribute to $\Fzero$, only $m\min\{d(2d+1),4m\}\leq2 m \kappa$ of them must actually be measured. 

The second auxiliary lemma bounds the deviation of $\Festzero$ from $\Fzero$ in terms of the errors made in the estimation of the individual expectation values entering $\Fzero$. 
\begin{lemma}[Stability of the Gaussian fidelity bound]
  \label{lem:stability}
  Let $\Festzero$ be defined like $\Fzero$ in \eq~\eqref{eq:boundexplicit} but with $\bfgamma$ and $\bfGamma^{(1)}$ replaced by $\bfgamma^\ast$ and $\bfGamma^{{(1)}\ast}$ and let $\epsilonmax \coloneqq \norm{\bfgamma-\bfgamma^\ast}_{\max{}}$ and $\varepsilonmax^{(1)} \coloneqq \norm{\bfGamma^{(1)}-\bfGamma^{{(1)}\ast}}_{\max{}} $. Then
  \begin{equation}
    \label{eq:stability_proof}
    |\Fzero - \Festzero| \leq 
    2 \smax^2 \left(\varepsilonmax^{(1)} \sqrt{\kappa}m + \epsilonmax \norm{\vec x}_2\sqrt{2m} \right) .
  \end{equation}
\end{lemma}
\begin{proof}
For convenience, we define the \emph{error vector}
    \begin{align}
      \label{eq:def_error_vector}
      \bfepsilon &\coloneqq \bfgamma-\bfgamma^\ast\in \RR^{2m} \\
            \intertext{and the \emph{error matrix}}
\label{eq:def_error_matrix}
      \Epsilon^{(1)} &\coloneqq \bfGamma^{(1)}-\bfGamma^{{(1)}\ast} .
    \end{align}
    The fidelity estimation error can then be written as 
  \begin{align}
    \label{eq:error_in_terms_of_O_D_x}
    \Fzero - \Festzero 
    &= \Tr\big[\O \D^{-2} \O^T (\Epsilon^{(1)}+2\bfepsilon\vec{x}^T)\big]  .
  \end{align}
  Due to H\"older's inequality, 
  \begin{align}
    | \Fzero - \Festzero |
\nonumber
&\leq \norm{\O{\D^{-2}}\O^T}_{\infty}\norm{\Epsilon^{(1)}+2\bfepsilon\vec{x}^T}_1 \\
&\leq \norm{\D^{-2}}_{\infty}\left( \norm{\Epsilon^{(1)}}_1+2 \norm{\bfepsilon}_2 \norm{\vec{x}}_2\right)  ,
\label{eq:ineq_in_pf_norms_L}
\end{align}
where in the last step we have used the bound \eqref{eq:norm_of_rank_one_matrix}. 
The second inequality in \eq~\eqref{eq:norm_ineqs} implies that $\norm{\bfepsilon}_2 \leq \sqrt{2m} \norm{\bfepsilon}_\infty$.
It remains to bound $\norm{\Epsilon}_1$. 
To this end, we use the first inequality in \eq~\eqref{eq:norm_ineqs} and \eq~\eqref{eq:2norm_property} to arrive at 
\begin{equation}
 \norm{\Epsilon^{(1)}}_1 \leq \sqrt{2m} \norm{\vect(\Epsilon^{(1)})}_2  .
\end{equation}
According to Lemma~\ref{lem:sparsity}, $\Fzero$ depends on at most $2 \kappa m$ entries of $\Epsilon^{(1)}$.
Without loss of generality we can hence omit all other elements of $\Epsilon^{(1)}$ and thus take $\vect(\Epsilon^{(1)})$ as a vector with at most $2 \kappa m$ elements.
Using this fact and the second inequality in \eq~\eqref{eq:norm_ineqs} we obtain 
\begin{align}
 \norm{\Epsilon^{(1)}}_1 
\nonumber
 &\leq \sqrt{2m} \sqrt{2m\kappa} \norm{\vect(\Epsilon^{(1)})}_\infty
 \\
 &= 2m\sqrt{\kappa} \norm{\Epsilon^{(1)}}_{\max{}}  , 
\end{align}
where we have used \eq~\eqref{eq:max_norm_property} in the last equality.
Finally, putting everything together and using that, by definition, $\norm{\D^{-2}}_{\infty}=\smax^2$, we arrive at the inequality \eqref{eq:stability_proof}.
\end{proof}

The third auxiliary lemma shows that the estimate of the fidelity lower-bound for target states $\rhot\in\CG$ obtained with the measurement scheme $\MG$ in Box~\ref{box:measurement_scheme} is reliable. This lemma is potentially interesting in its own right in scenarios other than certification.
\begin{lemma}[Reliable estimation of the Gaussian fidelity bound]
\label{lem:fidelity_bound_estimation}
Let $\alpha \in (0,1/2]$ and $\varepsilon>0$. Let $\Festzero$ be defined like $\Fzero$ in \eq~\eqref{eq:boundexplicit} but with $\bfgamma$ and $\bfGamma^{(1)}$ replaced by $\bfgamma^\ast$ and $\bfGamma^{{(1)}\ast}$, where $\bfgamma^\ast$ and $\bfGamma^{{(1)}\ast}$ are obtained as described by $\MG$ from  
\begin{equation} 
\label{eq:total_number_of_samples_gaussian_case}
  \ns = 2m \ns_1 + 2\kappa m \ns_2 
\end{equation}
copies of $\rhop$ 
, with $\ns_1$ and $\ns_2$ integers such that
\begin{subequations}
\label{eq:both_ns_bounds}
\begin{align}
 \ns_1 &\geq 2^6
 \frac{\sigma_1^2 (2m + 1)\,m\, \smax^4\, \norm{\vec x}^2_2}{\varepsilon^2 \ln\left(\tfrac{1}{1-\alpha} \right)}
 \label{eq:bound_for_cs1}
 \intertext{and}
 \ns_2 &\geq 2^5 
 \frac{\sigma_2^2 (2 \kappa m + 1)\, m^2\, \smax^4 \, \kappa}{\varepsilon^2 \ln\left(\tfrac{1}{1-\alpha} \right)}. 
 \label{eq:bound_for_cs2}
\end{align}
\end{subequations}
Then, 
\begin{equation}
\label{eq:statement_new_lemma}
  \PP\left[ |\Fzero - \Festzero| \leq \varepsilon\right]\geq 1-\alpha.
\end{equation}
\end{lemma}
\begin{proof}
Our proof strategy is to show that, with probability at least $1-\alpha$, the $2m$ elements of $\bfgamma$ and the $2 m \kappa$ relevant elements of $\bfGamma^{(1)}$ are estimated within additive errors bounded as 
\begin{subequations}
\label{eq:both_epsilonmax_bound}
\begin{align}
 \epsilonmax& \leq  \epsilonmax^*\coloneqq\frac{\varepsilon}{4 \smax^2 \norm{\vec x}_2 \sqrt{2m}} 
 \label{eq:epsilonmax_bound}
\intertext{and}
\varepsilonmax^{(1)} &\leq \varepsilonmax^{*(1)}\coloneqq\frac{\varepsilon}{4  \smax^2 \sqrt{\kappa} m }.
\label{eq:varepsilonmax_bound}
\end{align}
\end{subequations}
If the inequalities \eqref{eq:both_epsilonmax_bound} are fulfiled, then, due to Lemma~\ref{lem:stability}, it holds that $|\Fzero-\Festzero|\leq\varepsilon$. 

Since all $2m$ estimates $\{\gamma^\ast_{l}\}_l$ are sample averages over independent copies of $\rhop$, the measurement outcomes to obtain the $\{\gamma^\ast_{l}\}_l$ are all independent random variables, for each $l$ described by the same probability distribution. Furthermore, by assumption, the variances of these variables are all upper-bounded by $\sigma_1$. Analogously, the measurement outcomes to obtain all  $2 m \kappa$ relevant estimates $\{\Gamma^{{(1)}\ast}_{l,l'}\}_{l,l'}$ are independent random variables with variances upper-bounded by $\sigma_2$ and described, for each $l$ and $l'$, by the same probability distribution. Hence, according to Lemma~\ref{lem:general_estimation}, with the choice $\ol \alpha = \sqrt{1-\alpha}$, 
taking
\begin{subequations}
\label{eq:both_ns_bounds_proofs}
\begin{align}
 \ns_1 &\geq 2
 \frac{\sigma_1^2 (2m + 1)}{{\epsilonmax^*}^2 \ln\left(\tfrac{1}{1-\alpha} \right)}
 \label{eq:bound_for_cs1_proof}
 \intertext{and}
 \ns_2 &\geq 2 
 \frac{\sigma_2^2 (2 \kappa m + 1)}{{\varepsilonmax^{*(1)}}^2 \ln\left(\tfrac{1}{1-\alpha} \right)} ,
 \label{eq:bound_for_cs2_proof}
\end{align}
\end{subequations}
is sufficient for both
\begin{subequations}
\label{eq:both_estimates_bounds}
\begin{equation}
  \PP\bigl[ \forall l : \ |\gamma^\ast_{l} - \gamma_{l}| \leq \epsilonmax^* \bigr]\geq\sqrt{1-\alpha}
\end{equation}
and
\begin{equation}
\PP\bigl[ \forall\ \Gamma^{(1)}_{l,l'} \text{ relevant}: \ |\Gamma^{{(1)}\ast}_{l,l'} - \Gamma^{(1)}_{l,l'}| \leq  \varepsilonmax^{*(1)} \bigr]\geq\sqrt{1-\alpha} .
\end{equation}
\end{subequations}
Since the $\{\gamma^\ast_{l}\}_l$ and the $\{\Gamma^{{(1)}\ast}_{l,l'}\}_{l,l'}$ are independent random variables, \eqs~\eqref{eq:both_estimates_bounds} imply that
\begin{align}
\label{eq:las_implication}
\PP\left[
  \begin{array}{rcl}
    \forall l : \ |\gamma^\ast_{l} - \gamma_{l}| &\leq& \epsilonmax^* \\
    \text{and}\ \forall\ \Gamma^{(1)}_{l,l'} \text{ relevant}: \ |\Gamma^{{(1)}\ast}_{l,l'} - \Gamma^{(1)}_{l,l'}| &\leq& \varepsilonmax^{*(1)} \\
  \end{array}
  \right] \geq 1-\alpha .
\end{align}
Finally, inserting the definitions \eqref{eq:both_epsilonmax_bound} of $\epsilonmax^*$ and $\varepsilonmax^{*(1)}$ into \eqs~\eqref{eq:both_ns_bounds_proofs}, we see that \eqs~\eqref{eq:both_ns_bounds} are equivalent to \eqs~\eqref{eq:both_ns_bounds_proofs}.
\end{proof}

Now, we prove the theorem on quantum certification of Gaussian states.
\begin{proof}[Proof of Theorem~\ref{thm:certification_gaussian}]
That the total number of copies of $\rhop$ (see \eq~\eqref{eq:total_number_of_samples_gaussian_case}) needed for the certification test is asymptotically upper-bounded by \eq~\eqref{eq:numb_copies_bound_Gaussian} can be verified by straightforward calculation using \eqs~\eqref{eq:both_ns_bounds}. It remains to show that
 $(i)$ if $\rhop = \rhot$, then $\mathcal{T}$ accepts with probability at least $1-\alpha$, i.e.,  
\begin{align}
\label{eq:to_prove_1}
\PP\left[\Festzero\geq F_T+\varepsilon\right]\geq 1-\alpha,
\end{align}
 and $(ii)$ if $\rhop$ is such that $F < F_T$, then $\mathcal{T}$ rejects  with probability at least $1-\alpha$, i.e.,
\begin{align}
\label{eq:to_prove_2}
\PP\left[\Festzero < \FT + \varepsilon\right]\geq1-\alpha.
\end{align}

To show $(i)$, we first recall that, if $\rhop=\rhot$, $\Fzero = 1$. With this, \eq~\eqref{eq:statement_new_lemma} in Lemma~\ref{lem:fidelity_bound_estimation} implies that 
\begin{align}
\label{eq:prove_1}
\PP\left[ \Festzero \geq 1- \varepsilon\right]\geq 1-\alpha.
\end{align}
Since, by assumption of the theorem, the total estimation error is such that $\varepsilon\leq\frac{1-\FT}{2}$, it holds that $1-\varepsilon\geq F_T+\varepsilon$. Substituting the latter inequality into \eq~\eqref{eq:prove_1} yields \eq~\eqref{eq:to_prove_1}. 

To show $(ii)$, we first note that, since $\Fzero\leq F$ for all $\rhop$, if $F < F_T$, then 
\begin{align}
\label{eq:prove_2_0}
\Fzero<\FT.
\end{align} 
On the other hand, \eq~\eqref{eq:statement_new_lemma} implies also that 
\begin{align}
\label{eq:prove_2}
\PP\left[ \Festzero\leq\Fzero+\varepsilon\right]\geq 1-\alpha.
\end{align}
Inserting \eq~\eqref{eq:prove_2_0} into \eq~\eqref{eq:prove_2} yields \eq~\eqref{eq:to_prove_2}.
\end{proof}

\subsection{Proof of Theorem~\ref{thm:certification_1n}}
\label{sec:ProofsNonGaussian}
We proceed analogously to the last section and present three auxiliary lemmas specific to the fidelity bound $\Fone$ for the linear-optical case before proving Theorem~\ref{thm:certification_1n}.

To state the first lemma in a compact form we introduce the shorthand $\bfGamma \coloneqq \bigl(\bfGamma^{(i)}\bigr)_{i = 1, \dots, n+1}$ for the collection of all the moment tensors $\bfGamma^{(i)}$.
Analogously, the collection of all the estimates $\bfGamma^{(i)\ast}$ of the moment tensors, defined in Box~\ref{box:measurement_scheme_1n}, is denoted by $\bfGamma^{\ast} \coloneqq \bigl(\bfGamma^{(i)\ast}\bigr)_{i = 1, \dots, n+1}$. 

\begin{lemma}[Sparsity of the linear-optical fidelity bound]
\label{lem:sparsity_1n}
The fidelity bound $\Fone$ defined in \eq~\eqref{eq:boundexplicitnonGaussian2} can be written as
\begin{align}
\label{eq:sum_over_i}
\Fone = 1- \sum_{j=0}^{n}(-1/2)^{n-j}f_j\left(\bfGamma^{(j)},\bfGamma^{(j+1)}\right),
\end{align} 
where, for each $j \in \{0,\dots,n\}$, $f_j$ is a linear functional
given by
\begin{equation}
\label{eq:f_def}
f_j\left(\bfGamma^{(j)},\bfGamma^{(j+1)}\right) \coloneqq
\sum_{\mu=1}^{\binom{n}{j}} 
\Biggr\{
\Tr\biggl[ \Bigl(\1 \otimes \bigotimes_{i \in \Omega^{(j)}_\mu} \P^{(i)}\Bigr) \bfGamma^{(j+1)}\biggr] 
+ \frac{m+2n}{2}\Tr\biggl[ \Bigl(\bigotimes_{i \in \Omega^{(j)}_\mu} \P^{(i)}\Bigr) \bfGamma^{(j)}\biggr]
\Biggl\} .
\end{equation}
For each $j$, the functional $f_{j}$ depends on at most $\binom n j (2d)^{2j}$ elements of $\bfGamma^{(j)}$ and on at most $\binom n j 2m(2d)^{2j}$ elements of $\bfGamma^{(j+1)}$.
We call these the \emph{relevant elements for $f_j$}.
Moreover, $\Fone$ depends on at most 
\begin{equation}
  \label{eq:total_number}
  \nCor{\leq 2(n+1)} \coloneqq (1 + 2m) (4d^2 + 1)^n\in\landauO\left(m\left(4d^2+1\right)^n\right)
\end{equation}
elements of $\bfGamma$.
We call these the \emph{relevant elements} of $\bfGamma$.
\end{lemma}

The subindex ``$\leq 2(n+1)$'' in $\nCor{\leq 2(n+1)}$ makes reference to the fact that $2j$-th moments with $j\in[n+1]$ are taken into account.

\begin{proof}
\Eqs~\eqref{eq:sum_over_i} and \eqref{eq:f_def} can be checked by a straightforward calculation.
We use again the sparsity of $\O$, i.e., the property that its columns $\o^{(2j-1)}$ and $\o^{(2j)}$ have at least $2(m-d)$ zero element in common. 
Hence, each of the symmetric matrices $\P^{(j)}$, defined in \eq~\eqref{Projector}, has at most $(2d)^2$ non-zero elements. 
Consequently, the projectors $\bigotimes_{i \in \Omega^{(j)}_\mu} \P^{(i)}$ and $\1 \otimes \bigotimes_{i \in \Omega^{(j)}_\mu} \P^{(i)}$ in \eq~\eqref{eq:f_def} have at most $(2d)^{2j}$ and $2m(2d)^{2j}$ non-zero elements.
This implies that the first trace inside the sum in \eq~\eqref{eq:f_def} depends on at most $2m(2d)^{2j}$ elements of $\bfGamma^{(j+1)}$ and the second trace inside the sum on at most $(2d)^{2j}$ elements of $\bfGamma^{(j)}$. 
Hence, each $f_j$ depends on at most $\binom n j (2d)^{2j}$ elements of $\bfGamma^{(j)}$ and on at most $\binom n j 2m(2d)^{2j}$ elements of $\bfGamma^{(j+1)}$.
This proves the statements on the sparsity of the functionals $f_i$.
From this, it follows that $\Fone$ depends on at most 
  \begin{equation}
    \sum_{i=0}^n \left( \binom n i (2d)^{2i} + \binom n i 2m(2d)^{2i} \right) = (1 + 2m) (4d^2 + 1)^n 
  \end{equation}
  elements of $\bfGamma$ in total, where in the last step we have used the binomial theorem.
\end{proof}
It is important to mention that, as in Lemma~\ref{lem:sparsity} for the Gaussian case, the symmetry \eqref{eq:symmetry_property} of each $\bfGamma^{(j)}$ was not taken into account. Thus, even though the lemma gives the maximal total number of relevant elements that contribute to $\Fone$, many of them are not independent and must therefore not be measured.

The second auxiliary lemma upper-bounds the deviation of $\Festone$ from $\Fone$ in terms of the errors made in the estimation of the expectation values entering $\Fone$.
\begin{lemma}[Stability of the linear-optical fidelity bound]
\label{lem:stability_1n}
Let $\Festone$ be defined like $\Fone$ in \eq~\eqref{eq:boundexplicitnonGaussian2} but with $\bfGamma$ replaced by $\bfGamma^{\ast}$ and let $\varepsilonmax \coloneqq \norm{\bfGamma-\bfGamma^\ast}_{\max{}}$.
Then 
\begin{equation}
  \label{eq:Fone_norm_bound}
  |\Fone - \Festone| \leq\varepsilonmax \left(n+5m/2\right) \left(1/2+2d\sqrt{2nm}\right)^n .
\end{equation}
\end{lemma}

\begin{proof}
For convenience, we define, for each $j \in [n]$, the \emph{error tensor}
  \begin{equation}
    \Epsilon^{(j)} \coloneqq  \Gamma^{(j)}-\Gamma^{(j)\ast} \in \left(\RR^{2m\times 2m}\right)^{\otimes j} .
  \end{equation}
Using \eq~\eqref{eq:sum_over_i} and the fact that $f_j$ is linear, we write the fidelity estimation error as 
\begin{align}
\label{eq:fi_epsilon}
\Fone - \Festone
&=\sum_{j=0}^{n} (-1/2)^{n-j} f_j\left(\Epsilon^{(j)},\Epsilon^{(j+1)}\right),
\end{align}
Applying H\"older's inequality and using that the Schatten $\infty$-norm of a tensor product of projectors is bounded by $1$ yields
\begin{align}
\nonumber
|f_j\left(\Epsilon^{(j)},\Epsilon^{(j+1)}\right)|
& \leq
\binom{n}{j} 
\Bigl(\normb{\tilde{\Epsilon}^{(j+1)}}_1 +
\frac{m+2n}{2}\normb{\tilde{\Epsilon}^{(j)}}_1 \Bigr), 
\end{align}
where the matrix $\tilde{\Epsilon}^{(j)}$ is defined element-wise by $\tilde{\Epsilon}^{(j)}_{\vec{k}^{(j)},\vec{l}^{(j)}} \coloneqq \Epsilon^{(j)}_{k_1,l_1,\dots, k_j, l_j}$, where $\vec{k}^{(j)} \coloneqq (k_1,\dots , k_j)$ and $\vec{l}^{(j)} \coloneqq (l_1,\dots , l_j)$.
Thanks to the first bound in \eq~\eqref{eq:norm_ineqs} and \eq~\eqref{eq:2norm_property}, we arrive at
\begin{equation}
\label{eq:f_i}
|f_j\left(\Epsilon^{(j)},\Epsilon^{(j+1)}\right)|
\leq
\binom{n}{j} (2m)^{j/2} 
\Bigl(\sqrt{2m} \normb{\vect(\tilde{\Epsilon}^{(j+1)})}_2 + \frac{m+2n}{2}\normb{\vect(\tilde{\Epsilon}^{(j)})}_2 \Bigr)  . 
\end{equation}
According to Lemma~\ref{lem:sparsity_1n}, $f_j$ depends on at most $\binom n j 2m(2d)^{2j}$ elements of $\tilde{\Epsilon}^{(j+1)}$ and on at most $\binom n j (2d)^{2j}$ of $\tilde{\Epsilon}^{(j)}$.
Without loss of generality we can hence omit, in \eq~\eqref{eq:f_i}, all other elements in $\tilde{\Epsilon}^{(j)}$ and $\tilde{\Epsilon}^{(j+1)}$ and thus take $\vect(\tilde{\Epsilon}^{(j)})$ and $\vect(\tilde{\Epsilon}^{(j+1)})$ as vectors with at most $\binom n j (2d)^{2j}$  and $\binom n j 2m(2d)^{2j}$ elements, respectively.
Then the second bound in \eq~\eqref{eq:norm_ineqs} yields 
\begin{equation}
\label{eq:last_bound}
|f_j\left(\Epsilon^{(j)},\Epsilon^{(j+1)}\right)|
\leq
{\binom{n}{j}}^{3/2} (2m)^{j/2} (2d)^{j}
\Bigl[2m \normb{\tilde{\Epsilon}^{(j+1)}}_{\max{}} +
\frac{m+2n}{2}\normb{\tilde{\Epsilon}^{(j)}}_{\max{}} \Bigr]  . 
\end{equation}
Next, from \eq~\eqref{eq:fi_epsilon}, it follows that
\begin{equation}
  \label{eq:fi_epsilon_preprefinal}
  |\Fone - \Festone|\leq\varepsilonmax\biggl[\sum_{j=0}^{n}{\binom{n}{j}}^{3/2}\left(1/2\right)^{n-j} \left(\sqrt{2m}2d\right)^{j} \times \left(5m/2+n\right) \biggr].
\end{equation}
Finally, using $\binom{n}{j}^{1/2} \leq n^{j/2}$ and the binomial formula, we obtain the inequality \eqref{eq:Fone_norm_bound}.
\end{proof}

The third auxiliary lemma shows that the estimate of the fidelity lower-bound for target states $\rhot\in\CLON$ obtained with the measurement scheme $\MLON$ in Box~\ref{box:measurement_scheme_1n} is reliable. This lemma is potentially interesting in its own right in scenarios other than certification.
\begin{lemma}[Reliable estimation of the linear-optical fidelity bound]
\label{lem:fidelity_bound_estimation_1n}
Let $\alpha \in (0,1/2]$ and $\varepsilon>0$. Let $\Festone$ be defined like $\Fone$ in \eq~\eqref{eq:boundexplicitnonGaussian2} but with $\bfGamma$ replaced by $\bfGamma^{\ast}$, where $\bfGamma^{\ast}$ is obtained as described by $\MLON$ from  
\begin{equation}
  \label{eq:number_copies_NonGaussian}
  \ns = \nCor{\leq 2(n+1)}\ns_{\leq 2(n+1)} 
\end{equation}
copies of $\rhop$, with $\nCor{\leq 2(n+1)}$ an integer given by \eq~\eqref{eq:total_number} and $\ns_{\leq 2(n+1)}$ an integer given by
\begin{align} 
\label{eq:ns_1n}
 \ns_{\leq 2(n+1)} \geq \frac{\sigma_{\leq 2(n+1)}^2(\nCor{\leq2(n+1)}+1)}{\varepsilon^2  \ln(1/(1-\alpha))} \left(n+5m/2\right)^2 \left(1/2+2d\sqrt{2nm}\right)^{2n}.
\end{align}
Then, 
\begin{equation}
\label{eq:statement_new_lemma_1n}
  \PP\left[ |\Fone - \Festone| \leq \varepsilon\right]\geq 1-\alpha.
\end{equation}
\end{lemma}
\begin{proof}
Our proof strategy is similar to that of Lemma~\ref{lem:fidelity_bound_estimation}. That is, we show that, with probability at least $1-\alpha$, the $\nCor{\leq 2(n+1)}$ relevant elements of $\bfGamma$ are estimated  within additive errors bounded as
\begin{align}
\varepsilonmax &\leq \varepsilonmax^*\coloneqq\frac{\varepsilon}{\left(n+5m/2\right) \left(1/2+2d\sqrt{2nm}\right)^n}.
\label{eq:varepsilonmax_bound_1n}
\end{align}
If this inequality is fulfiled, then, due to Lemma~\ref{lem:stability_1n}, it holds that $|\Fone-\Festone|\leq\varepsilon$.

According to Lemma~\ref{lem:general_estimation}, with the choice $\ol \alpha = 1-\alpha$, 
taking
\begin{equation} 
\label{eq:ns_1n_proof}
 \ns_{\leq 2(n+1)} \geq \frac{\sigma_{\leq 2(n+1)}^2(\nCor{\leq2(n+1)}+1)}{{\varepsilonmax^*}^2  \ln(1/(1-\alpha))}.
\end{equation}
is sufficient to get 
\begin{equation}
\PP\left[ \forall\ \Gamma^{(i)}_{k_1,l_1,\dots , k_i, l_i} \text{ relevant}:
|\Gamma^{(i)\ast}_{k_1,l_1,\dots , k_i, l_i} - \Gamma^{(i)}_{k_1,l_1,\dots , k_i, l_i}| \leq \varepsilonmax^* \right]\geq1-\alpha
\end{equation}
Finally, inserting the definition \eqref{eq:varepsilonmax_bound_1n} of $\varepsilonmax^*$ into \eq~\eqref{eq:ns_1n_proof}, we see that \eq~\eqref{eq:statement_new_lemma_1n} is equivalent to \eq~\eqref{eq:ns_1n_proof}.
\end{proof}

Now, we prove the theorem on quantum certification of linear-optical network states.
\begin{proof}[Proof of Theorem~\ref{thm:certification_1n}]
The proof is analogous to the proof of Theorem~\ref{thm:certification_gaussian}, but with \eq~\eqref{eq:number_copies_NonGaussian}, \eq~\eqref{eq:numb_copies_bound}, $\Fone$, $\Festone$, Lemma~\ref{lem:fidelity_bound_estimation_1n}, and \eq~\eqref{eq:statement_new_lemma_1n} playing respectively the roles of \eq~\eqref{eq:total_number_of_samples_gaussian_case}, \eq~\eqref{eq:numb_copies_bound_Gaussian}
, $\Fzero$, $\Festzero$, Lemma~\ref{lem:fidelity_bound_estimation} and \eq~\eqref{eq:statement_new_lemma}.
\end{proof}

\subsection{Proof of Corollary~\ref{cor:certification_post_selected_G} }
\label{sec:proof_corollary_G}
The proof relies on three auxiliary lemmas equivalent to Lemmas~\ref{lem:sparsity}, \ref{lem:stability}, and \ref{lem:fidelity_bound_estimation}. 
\begin{lemma}[Sparsity of the locally post-selected Gaussian fidelity bound]
\label{lem:sparsity_PS}
$\FSyszero$ depends on at most $2 m \kappa$ elements of $\bfGamma_\Sys^{(1)}$.
We call these the \emph{relevant elements} of $\bfGamma_\Sys^{(1)}$.
\end{lemma}
\begin{proof}
The proof of the lemma is analogous to that of Lemma~\ref{lem:sparsity}.
 \end{proof}
\begin{lemma}[Stability of the locally post-selected Gaussian fidelity bound]
  \label{lem:stability_PS}
  Let $\FSysestzero$ be defined by the same expression to $\Fzero$ in \eq~\eqref{eq:boundexplicit} but divided by $\mathbb{P}(\bphi_\A|\rhot)$ and with $\bfgamma$ and $\bfGamma^{(1)}$ replaced by $\bfgamma_\Sys^\ast$ and $\bfGamma_\Sys^{{(1)}\ast}$, and let $\epsilonmax \coloneqq \norm{\bfgamma_\Sys-\bfgamma_\Sys^\ast}_{\max{}}$ and $\varepsilonmax^{(1)} \coloneqq \norm{\bfGamma_\Sys^{(1)}-\bfGamma_\Sys^{{(1)}\ast}}_{\max{}} $. Then
  \begin{equation}
    \label{eq:stability_proof_PS}
    |\FSyszero - \FSysestzero| \leq 
     \frac{2\smax^2}{\mathbb{P}(\bphi_\A|\rhot)} \left(\varepsilonmax^{(1)} \sqrt{\kappa}m + \epsilonmax \norm{\vec x}_2\sqrt{2m} \right) .
  \end{equation}
\end{lemma}
\begin{proof}
The proof of the lemma is similar to that of Lemma~\ref{lem:stability}, with the differences already explained in Section~\ref{sec:PS_measurement_scheme}. 
\end{proof}
\begin{lemma}[Reliable estimation of the locally post-selected Gaussian fidelity bound]
\label{lem:fidelity_bound_estimation_PS}
Let $\alpha \in (0,1/2]$ and $\varepsilon>0$. Let $\FSysestzero$ be defined by the same expression to $\FSyszero$ in \eq~\eqref{eq:boundexplicit} but divided by $\mathbb{P}(\bphi_\A|\rhot)$ and with $\bfgamma$ and $\bfGamma^{(1)}$ replaced by $\bfgamma_\Sys^\ast$ and $\bfGamma_\Sys^{{(1)}\ast}$, where $\bfgamma_\Sys^\ast$ and $\bfGamma_\Sys^{{(1)}\ast}$ are obtained as described in Section~\ref{sec:PS_measurement_scheme} from  
\begin{equation} 
\label{eq:total_number_of_samples_gaussian_case_PS}
  \ns = 2m \ns_1 + 2\kappa m \ns_2 
\end{equation}
copies of $\rhosysp$ 
, with $\ns_1$ and $\ns_2$ integers such that
\begin{subequations}
\label{eq:both_ns_bounds_PS}
\begin{align}
 \ns_1 &\geq 2^6
 \frac{\varsigma_1^2 (2m + 1)\,m\, \smax^4\, \norm{\vec x}^2_2}{\left[\mathbb{P}(\bphi_\A|\rhot)\varepsilon\right]^2 \ln\left(\tfrac{1}{1-\alpha} \right)}
 \label{eq:bound_for_cs1_PS}
 \intertext{and}
 \ns_2 &\geq 2^5 
 \frac{\varsigma_2^2 (2 \kappa m + 1)\, m^2\, \smax^4 \, \kappa}{\left[\mathbb{P}(\bphi_\A|\rhot)\varepsilon\right]^2 \ln\left(\tfrac{1}{1-\alpha} \right)}. 
 \label{eq:bound_for_cs2_PS}
\end{align}
\end{subequations}
Then, 
\begin{equation}
\label{eq:statement_new_lemma_PS}
  \PP\left[ |\FSyszero - \FSysestzero| \leq \varepsilon\right]\geq 1-\alpha.
\end{equation}
\end{lemma}
\begin{proof}
The proof of the lemma is analogous to that of Lemma~\ref{lem:fidelity_bound_estimation}. 
 \end{proof}
\begin{proof}[Proof of Corollary~\ref{cor:certification_post_selected_G}]
The proof is analogous to the proof of Theorem~\ref{thm:certification_gaussian} but with Lemmas~\ref{lem:sparsity_PS}, \ref{lem:stability_PS}, and \ref{lem:fidelity_bound_estimation_PS} playing respectively the roles of Lemmas~\ref{lem:sparsity}, \ref{lem:stability}, and \ref{lem:fidelity_bound_estimation}.
\end{proof}

\subsection{Proof of Corollary~\ref{cor:certification_post_selected_LON}}
\label{sec:proof_corollary_LON}
As in the previous subsection, the proof relies on three auxiliaryary lemmas equivalent to Lemmas~\ref{lem:sparsity_1n}, \ref{lem:stability_1n}, and \ref{lem:fidelity_bound_estimation_1n}. The proofs of the lemmas are analogous to, and follow immediately from, those of the latter. 
\begin{lemma}[Sparsity of the locally post-selected linear-optical fidelity bound]
\label{lem:sparsity_1n_PS}
The fidelity bound $\FSysone$, defined by the same expression as $\Fone$ in \eq~\eqref{eq:boundexplicitnonGaussian2} but divided by $\mathbb{P}(\bphi_\A|\rhot)$ and with $\bfGamma$ replaced by $\bfGamma_\Sys$,
can be written as
\begin{align}
\label{eq:sum_over_i_PS}
\FSysone = \frac{1}{\mathbb{P}(\bphi_\A|\rhot)}\left[1- \sum_{j=0}^{n}(-1/2)^{n-j}f_j\left(\bfGamma_\Sys^{(j)},\bfGamma_\Sys^{(j+1)}\right)\right],
\end{align} 
where, for each $j \in \{0,\dots,n\}$, $f_j$ is the same linear functional as in Lemma~\ref{lem:sparsity_1n}, defined by \eq~\eqref{eq:f_def}.
Moreover, $\FSysone$ depends on at most $\nCor{\leq 2(n+1)}$ elements of $\bfGamma_\Sys$, with $\nCor{\leq 2(n+1)}$ the same as in Lemma~\ref{lem:sparsity_1n} and given by \eq~\eqref{eq:total_number}.
\end{lemma}
\begin{proof}
The proof of the lemma is analogous to that of Lemma~\ref{lem:sparsity_1n}. 
\end{proof}

\begin{lemma}[Stability of the  locally post-selected  linear-optical fidelity bound]
\label{lem:stability_1n_PS}
Let $\FSysestone$ be defined by the same expression as $\Fone$ in \eq~\eqref{eq:boundexplicitnonGaussian2} but divided by $\mathbb{P}(\bphi_\A|\rhot)$ and with $\bfGamma$ replaced by $\bfGamma_\Sys^{\ast}$, and let $\varepsilonmax \coloneqq \norm{\bfGamma_\Sys-\bfGamma_\Sys^\ast}_{\max{}}$.
Then 
\begin{equation}
  \label{eq:Fone_norm_bound_PS}
  |\FSysone - \FSysestone| \leq\frac{\varepsilonmax}{\mathbb{P}(\bphi_\A|\rhot)} \left(n+5m/2\right) \left(1/2+2d\sqrt{2nm}\right)^n .
\end{equation}
\end{lemma}
\begin{proof}
The proof of the lemma is similar to that of Lemma~\ref{lem:stability_1n}, with the differences already explained in Section~\ref{sec:PS_measurement_scheme}. 
 \end{proof}
\begin{lemma}[Reliable estimation of the locally post-selected linear-optical fidelity bound]
\label{lem:fidelity_bound_estimation_1n_PS}
Let $\alpha \in (0,1/2]$ and $\varepsilon>0$. Let $\FSysestone$ be defined like $\Fone$ in \eq~\eqref{eq:boundexplicitnonGaussian2} but with $\bfGamma$ replaced by $\bfGamma_\Sys^{\ast}$, where $\bfGamma_\Sys^{\ast}$ is obtained as described in Section~\ref{sec:PS_measurement_scheme} from  
\begin{equation}
  \label{eq:number_copies_NonGaussian_PS}
  \ns = \nCor{\leq 2(n+1)}\ns_{\leq 2(n+1)} 
\end{equation}
copies of $\rhosysp$, with $\nCor{\leq 2(n+1)}$ an integer given by \eq~\eqref{eq:total_number} and $\ns_{\leq 2(n+1)}$ an integer given by
\begin{align} 
\label{eq:ns_1n_PS}
 \ns_{\leq 2(n+1)} \geq \frac{\varsigma_{\leq 2(n+1)}^2(\nCor{\leq2(n+1)}+1)}{\left[\mathbb{P}(\bphi_\A|\rhot)\,\varepsilon\right]^2  \ln(1/(1-\alpha))} \left(n+5m/2\right)^2 \left(1/2+2d\sqrt{2nm}\right)^{2n}.
\end{align}
Then, 
\begin{equation}
\label{eq:statement_new_lemma_1n_PS}
  \PP\left[ |\FSysone - \FSysestone| \leq \varepsilon\right]\geq 1-\alpha.
\end{equation}
\end{lemma}
\begin{proof}
The proof of the lemma is analogous to that of Lemma~\ref{lem:fidelity_bound_estimation_1n}. 
\end{proof}


\begin{proof}[Proof of Corollary~\ref{cor:certification_post_selected_LON}]
The proof is analogous to the proof of Theorem~\ref{thm:certification_1n} but with Lemmas~\ref{lem:sparsity_1n_PS}, \ref{lem:stability_1n_PS}, and \ref{lem:fidelity_bound_estimation_1n_PS} playing respectively the roles of Lemmas~\ref{lem:sparsity_1n}, \ref{lem:stability_1n}, and \ref{lem:fidelity_bound_estimation_1n}.
\end{proof}
\subsection{Proof of Theorem~\ref{thm:strong_certification}}
\label{Proofofthm:strong_certificationGaussian}
Crucial for the proof of this theorem is the expansion \eqref{eq:rhop_in_terms_rhot} of $\rhop$ in terms of $\rhot$ and $\rhot^{\perp}$, which leads to the definition \eqref{eq:def_phononmismatch} of the photon mismatch $\ntperp$. Also, before the proof, we note that the fidelity gap cannot be smaller than $\Delta \geq 2 \varepsilon$: 
The condition for acceptance of the test is  $\Festone \geq \FT+\Delta-\varepsilon$, whereas that for rejection is $\Festone < \FT+\varepsilon$. So, the threshold of acceptance, $F=\FT+\Delta-\varepsilon$, is not smaller than that of rejection, $F=\FT+\varepsilon$, iff $\Delta \geq 2 \varepsilon$.

\begin{proof}[Proof of Theorem~\ref{thm:strong_certification}] 
Theorems~\ref{thm:certification_gaussian} and \ref{thm:certification_1n} imply that $\rhop$ is rejected with probability at least $1-\alpha$ whenever $F < \FT$. Thus, it remains to show that if $\rhop$ is such that  $F\geq \FT + \Delta$, with $\Delta$ given by \eq~\eqref{eq:delta},
 then $\rhop$ is accepted with probability at least $1-\alpha$, i.e., that
\begin{equation}
\label{eq:robust_to_prove}
\PP\left[\Festone \geq \FT+\varepsilon\right] \geq 1-\alpha.
\end{equation}

So, let $F\geq \FT + \Delta$, with $\Delta$ given by \eqref{eq:delta}.
Using \eqs~\eqref{eq:eq:firstineqgeneric}, \eqref{eq:rhop_in_terms_rhot}, and \eqref{eq:def_phononmismatch}, we write $\Fone$ as 
\begin{equation}
\label{eq:bound_Delta}
\Fone = 1 - (1-F)\ntperp \geq 1-(1-(\FT+\Delta))\ntperp  .
\end{equation}
Using that
\begin{equation}
\label{eq:condondelta2}
\Delta\geq\frac{2\varepsilon+(1-\FT)(\tilde{n}^{\perp}-1)}{\tilde{n}^{\perp}}
\end{equation}
and inserting it into the inequality \eqref{eq:bound_Delta}, we obtain
\begin{equation}
\label{eq:Fnbound_in_robustness_proof}
\Fone \geq \FT+2\varepsilon  .
\end{equation} 
Finally, using \eqs~\eqref{eq:Fnbound_in_robustness_proof} and \eqref{eq:statement_new_lemma_1n}, we obtain 
\eq~\eqref{eq:robust_to_prove}.  
\end{proof}

\section{Number of measurement settings}
\label{sec:meas_settings}
In this section, we upper-bound the number of local measurement settings required for the estimation of our fidelity lower bounds. We do this explicitly only for the Gaussian and linear-optical network target states, the cases of the post-selected target states following immediately from them. 
\subsection{Gaussian case}
\label{sec:SettingsGaussian}
Here, we show that the $2 m d$ single-quadrature and the $m \kappa$ two-quadrature observables listed in Box~\ref{box:measurement_scheme}, required for the measurement scheme $\MG$, can all be measured using $m+3$ different experimental arrangements. We do this by explicitly describing a measurement strategy that features such a scaling. 

The two-body observables $\hat{q}_j\hat{q}_k$, $\hat{q}_j\hat{p}_k$, and $\hat{p}_j\hat{p}_k$, for $j\neq k$, can be measured by simultaneously homodyning modes $j$ and $k$. For all possible pairs of modes, this consumes $m+2$ different homodyne settings: A single setting $(\hat{q}_1,\hat{q}_2, \hdots , \hat{q}_m)$ for all the second moments of the form $\langle\hat{q}_j\hat{q}_k\rangle_{\rhop}$; another single setting $(\hat{p}_1,\hat{p}_2, \hdots , \hat{p}_m)$ for those of the form  $\langle\hat{p}_j\hat{p}_k\rangle_{\rhop}$; and the $m$ settings $(\hat{p}_1,\hat{q}_2, \hdots , \hat{q}_m)$, $(\hat{q}_1,\hat{p}_2, \hat{q}_3, \hdots , \hat{q}_m)$, $\hdots$, and $(\hat{q}_1, \hdots , \hat{q}_{m-1}, \hat{p}_m)$ for those of the form $\langle\hat{q}_j\hat{p}_k\rangle_{\rhop}$ and $\langle\hat{p}_j\hat{q}_k\rangle_{\rhop}$ with $j\neq k$. In addition, all the single-body observables $\hat{q}_j$, $\hat{p}_j$, $\hat{q}_j^2$, and $\hat{p}_j^2$, are measured also with these same settings. With this, 
we have accounted, so far, for all the first moments $\gamma_l$ and all the second moments $\Gamma^{(1)}_{l,l'}$ with $(l,l')\neq(2j-1,2j)$ for all $ j\in [m]$.

The remaining second moments, $\Gamma^{{(1)}}_{2j-1,2j}$ with $ j\in [m]$, correspond to the single-mode observables 
$(\hat{q}_j\hat{p}_j+\hat{p}_j\hat{q}_j)/2$. To measure these, Arthur can homodyne  each mode $j$ independently in the rotated quadrature $(\hat{q}_j+\hat{p}_j)/\sqrt{2}$. This requires a single setting: $\left[ (\hat{q}_1+\hat{p}_1)/\sqrt{2} , (\hat{q}_2+\hat{p}_2)/\sqrt{2}, \hdots , (\hat{q}_1+\hat{p}_1)/\sqrt{2}\right]$.
In this setting,  he can estimate all the moments of the form $\langle(\hat{q}_j+\hat{p}_j)^2/2\rangle_{\rhop}$.
The latter estimates, upon subtraction of $\langle\hat{q}_j^2\rangle_{\rhop}/2$ and 
$\langle \hat{p}_j^2\rangle_{\rhop}/2$, whose settings have already been accounted for, finally make it possible to calculate an estimate of $\langle(\hat{q}_j\hat{p}_j+\hat{p}_j\hat{q}_j)/2 \rangle_{\rhop}$, using the equation
\begin{equation}
\label{eq:relation_qp}
\frac{1}{2}(\hat{q}_j\hat{p}_j+\hat{p}_j\hat{q}_j)=\left(\frac{\hat{q}_j+\hat{p}_j}{\sqrt{2}}\right)^2-\frac{\hat{q}_j^2}{2}-\frac{\hat{p}_j^2}{2} .
\end{equation}
The last setting, plus the $m+2$ ones already accounted for in the previous paragraph, yields a total of $m+3$ different homodyne settings, as promised.  

Finally, a comment on the error estimation is in order. In any measurement strategy where moments are estimated indirectly, their errors must be obtained from those of the directly measured quantities via error propagation. For instance, in the strategy just described, the error of each $\Gamma^{{(1)}}_{2j-1,2j}$ needs to be calculated from those of  $\langle(\hat{q}_j+\hat{p}_j)^2/2\rangle_{\rhop}$, $\langle\hat{q}_j^2\rangle_{\rhop}$, and 
$\langle \hat{p}_j^2\rangle_{\rhop}$. 
This leads,  for each indirectly estimated moment, to an increase in the number of copies of $\rhop$ required to attain a given error.
Nevertheless, this usually has no impact on the leading terms of the total resource scaling of the protocol. For example, in the described strategy
the global scaling given in \eq~\eqref{eq:numb_copies_bound_Gaussian} 
remains unaltered.
\subsection{Linear-optical case}
\label{sec:Settingsnon-Gaussian}
Here, we show that the $\nCor{\leq 2(n+1)}\in\landauO\left(m\left(4d^2+1\right)^n\right)$ observables listed in Box~\ref{box:measurement_scheme_1n}, required for the measurement scheme $\MLON$, can all be measured using  at most  $\binom m n  2^{n+1}$ different experimental arrangements. As in the previous section, we do this by explicitly describing a measurement strategy that features the promised scaling.

The scheme $\MLON$ requires the measurement of products of an even number between $2$ and $2(n+1)$  quadrature operators. We describe the measurement strategy as follows. First, we upper-bound the number of homodyne settings required for the measurement of all possible products of $2n$ quadrature operators, necessary for estimating all $n$-th moments $\bfGamma^{(n)}$. The measurement of products of fewer quadrature operators can clearly be carried out with the same settings. Then, we show that the particular products of $2(n+1)$ quadrature operators that appear in $\MLON$, corresponding to the relevant elements of $\bfGamma^{(n+1)}$, do not require extra settings either.

Consider all products of $2n$ quadrature operators. Among these, we focus first on those containing exclusively either $\hat{q}_j$ or $\hat{p}_j$ (or powers thereof) for each $j$-th mode but exclude observables such as $(\hat{q}_j\hat{p}_j+\hat{p}_j\hat{q}_j)/2$, which we address in the next paragraph.
Let us divide this family into two subfamilies: $(i)$ those for which the number of operators $\hat{q}_j$ is smaller or equal than that of the operators $\hat{p}_j$ and $(ii)$ those for which the number of operators $\hat{q}_j$ is greater than that of the operators $\hat{p}_j$. All correlators in the subfamily $(i)$ can be measured with homodyne settings where $n$ modes are detected in the position quadrature $\hat{q}_j$ and the remaining $m-n$ ones in the momentum  quadrature $\hat{p}_j$. All those in the subfamily $(ii)$ can be measured with homodyne settings where $n$ modes are detected in momentum and the remaining $m-n$ ones in position. Taking the two subfamilies into account, there are at most $2 \binom m n$ different such settings.

Let us now focus on the products of $2n$ quadrature operators that include different quadrature operators on a same mode, such as $(\hat{q}_j\hat{p}_j+\hat{p}_j\hat{q}_j)/2$ (or powers thereof). At most, $n$ factors as $(\hat{q}_j\hat{p}_j+\hat{p}_j\hat{q}_j)/2$ can appear in each product of $2n$ quadrature operators. From \eq~\eqref{eq:relation_qp}, we know that by replacing in each of the settings for the subfamily $(i)$ above a quadrature $\hat{q}_j$ with the rotated quadrature $(\hat{q}_j+\hat{p}_j)/\sqrt{2}$, Arthur can indirectly estimate the expectation values of all the $2n$-quadrature products of the form: 
\begin{align}
\label{eq:indirect_est1}
\nonumber
&\frac{1}{2}(\hat{q}_j\hat{p}_j+\hat{p}_j\hat{q}_j)\\
\nonumber
&\times\ \text{up to $n-1$ position operators}\\
&\times\ \text{at least $n$ momentum operators}.
\end{align}
In turn, by replacing, in each of the resulting settings, a further quadrature $\hat{q}_{j'}$ with $(\hat{q}_{j'}+\hat{p}_{j'})/\sqrt{2}$, he can measure all the observables of the form
\begin{align}
\label{eq:indirect_est2}
\nonumber
&\frac{1}{2}(\hat{q}_j\hat{p}_j+\hat{p}_j\hat{q}_j)\times \frac{1}{2}(\hat{q}_{j'}\hat{p}_{j'}+\hat{p}_{j'}\hat{q}_{j'})\\
\nonumber
&\times\ \text{up to $n-2$ position operators}\\
&\times\ \text{at least $n$ momentum operators}.
\end{align}
Concatenating this procedure, he can measure all the $2n$-quadrature products where each mode contributes with either $\hat{q}_j\hat{p}_j+\hat{p}_j\hat{q}_j$, $\hat{q}_j$, or $\hat{p}_j$, and the number of operators $\hat{q}_j$ is smaller or equal than that of the operators $\hat{p}_j$.
Equivalently, by proceeding analogously with the subfamily $(ii)$ and the quadratures $\hat{p}_j$, he can measure all $2n$-quadrature products where each mode contributes with either $\hat{q}_j\hat{p}_j+\hat{p}_j\hat{q}_j$, $\hat{q}_j$, or $\hat{p}_j$, and the number of operators $\hat{q}_j$ is greater than that of the operators $\hat{p}_j$. This is enough to indirectly estimate  the expectation values of all $2n$-quadrature products. 
For each setting of the two subfamilies, $n$ modes can be rotated, giving rise to $2^n$ setting ramifications. Hence, taking into account all the settings of the two subfamilies and their ramifications, we count a total of at most $2 \binom m n  2^n=\binom m n  2^{n+1}$ different settings.
This counting clearly over-counts the necessary settings but is enough for our purposes.

Finally, we consider the products of $2(n+1)$ quadrature operators appearing in the relevant elements of $\bfGamma^{(n+1)}$. The $(n+1)$-th moment tensor $\bfGamma^{(n+1)}$ is special in that, in contrast to the lower-moment tensors, it appears in just the first of the two traces in \eqs~\eqref{eq:boundexplicitnonGaussian2} and \eqref{eq:f_def}. In particular, according to Box~\ref{box:measurement_scheme_1n}, $\Gamma^{(n+1)}_{k_1,l_1,\dots, k_n, l_n, k_{n+1}, l_{n+1}}$ is a relevant element of $\bfGamma^{(n+1)}$ if, and only if, $k_{n+1}=l_{n+1}$. This implies that the observables containing the factor $(\hat{q}_{n+1}\hat{p}_{n+1}+\hat{p}_{n+1}\hat{q}_{n+1})$ do not contribute to the relevant elements of $\bfGamma^{(n+1)}$, only those containing either $\hat{q}^2_{n+1}$ or $\hat{p}^2_{n+1}$ are relevant. Hence, the relevant $2(n+1)$ quadrature products are those composed of the $2n$ quadrature products relevant for $\bfGamma^{(n)}$ times either $\hat{q}^2_{n+1}$ or $\hat{p}^2_{n+1}$. Now, in each  setting of 
the two subfamilies of the previous paragraph, $2n$ modes are used to measure a $2n$-quadrature observable relevant for $\bfGamma^{(n)}$ and the other $m-n$ modes, which are all set either to position or momentum, are ignored. Thus, each relevant element of $\bfGamma^{(n+1)}$ can be estimated by not ignoring one out of the latter $m-n$ modes. That is, the settings to estimate the $2n$-moments $\bfGamma^{(n)}$ already cover also the estimation of $2(n+1)$-moments $\bfGamma^{(n+1)}$.
So, the total number of settings used throughout is at most $\binom m n  2^{n+1}$.

As in the end of the previous section, we make a final remark on the error estimation. Also here, the errors of the indirectly estimated moments must be obtained via error propagation, which leads again to an increase in the total number of copies of $\rhop$. 
Nevertheless, their global scaling with $n$ remains of the same order as that given in \eq~\eqref{eq:numb_copies_bound}
.

\section{Stability against systematic errors}
\label{sec:systematic_errors}
Apart from statistical errors, Arthur's measurement procedure could also have systematic errors. That is, if the characterisation of his single-mode measurement channels is erroneous, he could actually be measuring different observables from the ones he thinks he does. Theorems~\ref{thm:certification_gaussian} and \ref{thm:certification_1n}, as well as their Corollaries~\ref{cor:certification_post_selected_G} and \ref{cor:certification_post_selected_LON}, consider only statistical errors, i.e., those that can be decreased by increasing the number of measurement repetitions (and, hence, the number of copies of $\rhop$). Since systematic errors cannot be decreased by accumulating statistics, no certification method based exclusively on the measurement statistics can rule them out. However, the stability analyses of Lemmas~\ref{lem:stability}, \ref{lem:stability_1n}, \ref{lem:stability_PS} and \ref{lem:stability_1n_PS} hold regardless of the nature of errors. Thus, the experimental 
estimates $\Festzero$, $\Festone$, $\FSysestzero$, and $\FSysestone$ (and, therefore, also the certification tests)
turn out to be robust also against small systematic errors: The total fidelity deviations due to systematic errors scales linearly with the magnitude of the largest systematic error and polynomially in all the other relevant parameters as given in \eqs~\eqref{eq:stability_proof}, \eqref{eq:Fone_norm_bound}, \eqref{eq:stability_proof_PS}, and \eqref{eq:Fone_norm_bound_PS}.

Still, it is illustrative to consider a physically relevant example. A typical systematic error is non-unit quantum efficiency of the detectors used for homodyning. In that case, the probability density function $\tilde{\PPP}$ of measurement outcomes $r$ of a quadrature $\hat{r}$ equals the ideal one $\PPP$ convolutioned with the normal distribution $\mathcal{N}$ of mean zero and squared variance $(1-\eta)/4\eta$, where $\eta$ is the quantum efficiency of the detectors \cite{Ferraro}. That is, $\tilde{\PPP}(r)=(\PPP\ast\mathcal{N})(r)\coloneqq\int dr'\PPP(r')\mathcal{N}(r-r')$. Using that the first and second \emph{non-central moments} of $\mathcal{N}$ satisfy 
\begin{subequations}
\begin{align}
\langle r\rangle_{\mathcal{N}}&\coloneqq\int dr r \mathcal{N}(r-r')=r'\\
\intertext{and} 
 \langle r^2\rangle_{\mathcal{N}}&\coloneqq\int dr r^2 \mathcal{N}(r-r')=r'^2+\frac{1-\eta}{4\eta},
\end{align}
\end{subequations}
respectively, one obtains that 
\begin{subequations}
\begin{align}
\langle r\rangle_{\tilde{\PPP}}&=\langle r\rangle_{\PPP}\\
\intertext{and} 
\langle r^2\rangle_{\tilde{\PPP}}&=\langle r^2\rangle_{\PPP}+\frac{1-\eta}{4\eta}.
\end{align}
\end{subequations}
That is, the expectation value of $\hat{r}$ is not affected by this type of systematic errors and that of $\hat{r}^2$ deviates from the ideal one by $(1-\eta)/(4\eta)$. 
Furthermore, the expectation values of products of quadrature operators acting on different modes are also not affected, as this type of systematic error acts independently on different modes. 

In the absence of statistical errors, this leads to an error vector $\bfepsilon=\mathbf{0}$ and an error matrix $\Epsilon^{(1)}$ that is diagonal and such that $\norm{\Epsilon^{(1)}}_{\max{}}\leq({1-\eta})/{(4\eta)}$, 
so that $\norm{\Epsilon^{(1)}}_{1}\leq m ({1-\eta})/{(2\eta)}$. 
Inserting this into \eq~\eqref{eq:ineq_in_pf_norms_L}, we see for instance that, for Gaussian targets, the contribution to the deviation of the fidelity estimate due to non-ideal detector efficiency in the homodyne detectors is smaller than $\smax^2m\frac{1-\eta}{2\eta}$.
This, in turn, is smaller or equal than a desired constant maximal error $\varepsilon$ if
\begin{equation}
  \label{eq:scaling_eta}
  \eta\geq\frac{\smax^2 m}{2\varepsilon+ \smax^2 m} \approx 1-\frac{2 \varepsilon}{\smax^2 m},
\end{equation}
where the approximation holds whenever $\smax^2 m\gg2\varepsilon$.
The scaling given by the bound \eqref{eq:scaling_eta} is experimentally convenient in that, in particular, it implies that the detector inefficiency $1-\eta$ needs to decrease only inversely proportional with the number of modes $m$.

Another typical systematic error is the limited power of the local oscillator field used for the homodyne detection: The homodyne (photocurrent difference) statistics, i.e., the distribution of homodyne measurement outcomes, match exactly the statistics of the corresponding quadrature  only in the limit of an intense local-oscillator beam \cite{Braunstein90}. The most obvious difference is that the homodyne statistics is discrete whereas the quadrature statistics is continuous, with the former approximating the latter increasingly better as the local-oscillator power increases. However, we emphasise that our method relies on the estimation of only the expectation values of quadratures and not their full statistics. It can be seen that, provided that the local oscillator is in a coherent state, the effect of limited power is just to increase the variance of the effective quadrature without changing its expectation value with respect to the ideal case. Furthermore, in the multi-mode scenario, if the 
different modes are homodyned with independent local oscillators, the latter is also true for products of quadratures, as the ones considered in this work. Therefore, the effect of systematic errors due to limited homodyne local-oscillator power in our fidelity estimates is expected not to be critical either.

\section{Auxiliary mathematical relations}
\label{sec:auxiliary_maths}

\subsection{Derivation of the properties of the operator valued Pochhammer-Symbol}
\label{sec:proof_of_eq:general_fidelity_bound}
We begin with \eq~\eqref{eq:adaggera}. The general relationship
\begin{equation}
\label{eq:appendix:adaggera}
 (a\ad_j)^t\hat{n}_j(a_j)^{t} = p_t(\hat{n}_j),
\end{equation}
for $t \in \NN$, can be shown by induction starting from $p_{0}(\hat{n}_j)=\hat{n}_j$ and noting that, for all $t\geq -1$, 
\begin{align}
\nonumber 
 a\ad_j  p_{t}(\hat n)  a_j &= a\ad_j  \hat n_j  (\hat n_j - 1)  (\hat n_j - 2 )  \cdots  (\hat n_j -t)  a_j\\
\nonumber 
 &= a\ad_j  \hat n_j  (\hat n_j - 1)  (\hat n_j - 2 )  \cdots  (\hat n_j -(t-1)) a_j  (\hat n_j -(t+1))\\
\nonumber
 &= p_{t}(\hat n_j)  (\hat n_j -(t+1)) \\
 &= p_{t+1}(\hat n_j) , 
\end{align}
as can be verified using the commutation relations between $a_j$ and $a\ad_j$. Setting $t=n_j$ gives \eq~\eqref{eq:adaggera}
.

In turn, \eq~\eqref{eq:adaggerb} 
can be shown by noting that 
\begin{align}
\label{eq:appendix:adaggerb}
(a\ad_j)^{n_j}(a_j)^{n_j}=(a\ad_j)^{n_j-1} \hat{n}_j(a_j)^{n_j-1}
\end{align}
and applying \eq~\eqref{eq:appendix:adaggera}, for $t=n_j-1$, to the right-hand side of \eqref{eq:appendix:adaggerb}.

\subsection{Proof of the bound \texorpdfstring{\eqref{eq:bound_for_cs}}{}}
\label{sec:app_bound}
Note that for $x=0$ both sides of \eq~\eqref{eq:bound_for_cs} yield $1$ and hence the bound holds in that case. 
We make the substitution $y = 1 /x$ and show that the bound \eqref{eq:bound_for_cs} holds for all $x>0$ by proving the following: 
\begin{equation} \label{eq:bound_for_cs_y}
 \frac 1 {1-\e^{-y}} \leq \frac 1 y + \frac 1 {2(1+1/y)} + \frac 1 2 \quad \forall y\geq 0  .
\end{equation}
But this is equivalent to 
\begin{equation}
 2 y^2 + 3 y + 2 \leq \e^y(2 + y) . 
\end{equation}
A straight forward calculation shows that both sides and also the first derivatives of both sides coincide at $y=0$, while the second derivative of the right hand side is always larger than the second derivative of the left hand side. This proves \eq~\eqref{eq:bound_for_cs_y} and hence finishes the proof of the bound \eqref{eq:bound_for_cs}.

\end{document}